\PassOptionsToPackage{unicode}{hyperref}
\PassOptionsToPackage{hyphens}{url}
\PassOptionsToPackage{dvipsnames,svgnames,x11names}{xcolor}
\documentclass[
  12pt]{article}

\usepackage{amsmath,amssymb}
\usepackage{iftex}
\ifPDFTeX
  \usepackage[T1]{fontenc}
  \usepackage[utf8]{inputenc}
  \usepackage{textcomp} 
\else 
  \usepackage{unicode-math}
  \defaultfontfeatures{Scale=MatchLowercase}
  \defaultfontfeatures[\rmfamily]{Ligatures=TeX,Scale=1}
\fi
\usepackage{lmodern}
\ifPDFTeX\else  
\fi
\IfFileExists{upquote.sty}{\usepackage{upquote}}{}
\IfFileExists{microtype.sty}{
  \usepackage[]{microtype}
  \UseMicrotypeSet[protrusion]{basicmath} 
}{}
\makeatletter
\@ifundefined{KOMAClassName}{
  \IfFileExists{parskip.sty}{%
    \usepackage{parskip}
  }{
    \setlength{\parindent}{0pt}
    \setlength{\parskip}{6pt plus 2pt minus 1pt}}
}{
  \KOMAoptions{parskip=half}}
\makeatother
\usepackage{xcolor}
\makeatletter
\ifx\paragraph\undefined\else
  \let\oldparagraph\paragraph
  \renewcommand{\paragraph}{
    \@ifstar
      \xxxParagraphStar
      \xxxParagraphNoStar
  }
  \newcommand{\xxxParagraphStar}[1]{\oldparagraph*{#1}\mbox{}}
  \newcommand{\xxxParagraphNoStar}[1]{\oldparagraph{#1}\mbox{}}
\fi
\ifx\subparagraph\undefined\else
  \let\oldsubparagraph\subparagraph
  \renewcommand{\subparagraph}{
    \@ifstar
      \xxxSubParagraphStar
      \xxxSubParagraphNoStar
  }
  \newcommand{\xxxSubParagraphStar}[1]{\oldsubparagraph*{#1}\mbox{}}
  \newcommand{\xxxSubParagraphNoStar}[1]{\oldsubparagraph{#1}\mbox{}}
\fi
\makeatother

\usepackage{longtable,booktabs,array}
\usepackage{calc} 
\usepackage{etoolbox}
\makeatletter
\patchcmd\longtable{\par}{\if@noskipsec\mbox{}\fi\par}{}{}
\makeatother
\IfFileExists{footnotehyper.sty}{\usepackage{footnotehyper}}{\usepackage{footnote}}
\makesavenoteenv{longtable}
\usepackage{graphicx}
\makeatletter
\def\maxwidth{\ifdim\Gin@nat@width>\linewidth\linewidth\else\Gin@nat@width\fi}
\def\maxheight{\ifdim\Gin@nat@height>\textheight\textheight\else\Gin@nat@height\fi}
\makeatother
\setkeys{Gin}{width=\maxwidth,height=\maxheight,keepaspectratio}
\makeatletter
\def\fps@figure{htbp}
\makeatother

\makeatletter
\@ifpackageloaded{caption}{}{\usepackage{caption}}
\AtBeginDocument{%
\ifdefined\contentsname
  \renewcommand*\contentsname{Table of contents}
\else
  \newcommand\contentsname{Table of contents}
\fi
\ifdefined\listfigurename
  \renewcommand*\listfigurename{List of Figures}
\else
  \newcommand\listfigurename{List of Figures}
\fi
\ifdefined\listtablename
  \renewcommand*\listtablename{List of Tables}
\else
  \newcommand\listtablename{List of Tables}
\fi
\ifdefined\figurename
  \renewcommand*\figurename{Figure}
\else
  \newcommand\figurename{Figure}
\fi
\ifdefined\tablename
  \renewcommand*\tablename{Table}
\else
  \newcommand\tablename{Table}
\fi
}
\@ifpackageloaded{float}{}{\usepackage{float}}
\floatstyle{ruled}
\@ifundefined{c@chapter}{\newfloat{codelisting}{h}{lop}}{\newfloat{codelisting}{h}{lop}[chapter]}
\floatname{codelisting}{Listing}

\makeatother
\makeatletter
\@ifpackageloaded{caption}{}{\usepackage{caption}}
\@ifpackageloaded{subcaption}{}{\usepackage{subcaption}}
\makeatother

\ifLuaTeX
  \usepackage{selnolig}  
\fi
\usepackage[]{natbib}
\usepackage{bookmark}

\IfFileExists{xurl.sty}{\usepackage{xurl}}{} 
\hypersetup{
  pdftitle={Title},
  pdfauthor={Author 1; Author 2},
  pdfkeywords={3 to 6 keywords, that do not appear in the title},
  colorlinks=true,
  linkcolor={blue},
  filecolor={Maroon},
  citecolor={Blue},
  urlcolor={Blue},
  pdfcreator={LaTeX via pandoc}}

\newcommand{\anon}{1}

\usepackage{amsfonts,amssymb,mathtools,flexisym}
\usepackage{amsthm}
\usepackage{float}
\usepackage{xcolor}
\usepackage{bbm}
\usepackage{appendix}
\usepackage{hyperref}
\usepackage{csquotes}
\usepackage{bm}
\usepackage{multirow}
\usepackage{mathrsfs}

\usepackage[section]{placeins}
\usepackage{float}
\usepackage{booktabs}
\usepackage{authblk} 

\usepackage{cancel}
\usepackage{comment}

\begin{document}

\def\spacingset#1{\renewcommand{\baselinestretch}%
{#1}\small\normalsize} \spacingset{1}

\newtheorem{theorem}{Theorem}
\newtheorem{proposition}{Proposition}
\newtheorem{corollary}{Corollary}
\newtheorem{assumption}{Assumption}
\newtheorem{remark}{Remark}
\renewcommand{\proofname}{\bf{Proof}}

\theoremstyle{definition}

\renewcommand\P{\mathbb{P}}
\newcommand\E{\mathbb{E}}
\renewcommand\H{\mathcal{H}}
\newcommand\Z{\mathbb{Z}}
\renewcommand\L{\mathcal{L}}
\renewcommand\S{\mathcal{S}}
\newcommand\N{\mathcal{N}}
\newcommand\G{\mathcal{G}}
\newcommand\C{\mathcal{C}}
\newcommand\D{\mathcal{D}}
\newcommand\cE{\mathcal{E}}
\newcommand\F{\mathcal{F}}
\newcommand\bN{\mathbb{N}}
\newcommand\R{\mathbb{R}}
\newcommand\A{\mathcal{A}}
\newcommand\B{\mathcal{B}}
\renewcommand\O{\mathcal{O}}
\newcommand\K{\mathcal{K}}
\newcommand\cT{\mathcal{T}}
\newcommand\I{\mathcal{I}}
\newcommand\cP{\mathcal{P}}
\newcommand\cR{\mathcal{R}}
\newcommand\Q{\mathcal{Q}}
\newcommand\M{\mathcal{M}}

\newcommand\bnu{\boldsymbol{\nu}}
\newcommand\hnu{\widehat{\boldsymbol{\nu}}}
\newcommand\bw{\boldsymbol{\omega}}
\newcommand\brho{\boldsymbol{\rho}}
\newcommand\hrho{\widehat{\boldsymbol{\rho}}}
\newcommand\Dd{\mathcal{D}^{\dagger}}
\newcommand\hD{\widehat{\mathcal{D}}}
\newcommand\hC{\widehat{\mathcal{C}}}
\newcommand\hE{\widehat{\mathcal{E}}}
\newcommand\bPhi{\boldsymbol{\Phi}}
\newcommand\bzeta{\boldsymbol{\zeta}}
\newcommand\tbPhi{\widetilde{\boldsymbol{\Phi}}}
\newcommand\indicator{\mathbbm{1}}
\newcommand\beps{\boldsymbol{\varepsilon}}
\newcommand\by{\boldsymbol{y}}
\newcommand\bz{\boldsymbol{z}}
\newcommand\bh{\boldsymbol{h}}

\newcommand\tr{{\rm tr}}
\def\T{\mathrm{\scriptscriptstyle T}}

\newcommand\smallo{
  \mathchoice
    {{\scriptstyle\O}}
    {{\scriptstyle\O}}
    {{\scriptscriptstyle\O}}
    {\scalebox{.7}{$\scriptscriptstyle\O$}}
}

\allowdisplaybreaks
\setlength{\parindent}{2em} 


\if1\anon
{
  \title{\bf Asymptotic Theory for Regularized Estimation in Functional Time Series Models}
  \author[1]{Ying Niu\thanks{These authors contributed equally to this work}} 
  \author[2]{Yuwei Zhao$^*$}
  \author[3]{Zhao Chen\thanks{Corresponding author: zchen\_fdu@fudan.edu.cn}} 
  \author[4]{Christina Dan Wang\thanks{Corresponding author: christina.wang@nyu.edu}}
  \affil[1]{\small Shanghai Center for Mathematical Sciences, Fudan University}
  \affil[2]{\small Department of Financial and Actuarial Mathematics, Xi'an Jiaotong-Liverpool University}
  \affil[3]{\small School of Data Science, Fudan University}
  \affil[4]{\small Business Division, New York University Shanghai}
  \maketitle
} \fi

\if0\anon
{
  \bigskip
  \bigskip
  \bigskip
  \begin{center}
    {\LARGE\bf Asymptotic Theory for Regularized Estimation in Functional Time Series Models}
\end{center}
  \medskip
} \fi

\bigskip
\begin{abstract}
Functional autoregressive (FAR) models provide a fundamental framework for analyzing temporally dependent functional data. However, the infinite-dimensional nature of the underlying Hilbert space introduces intrinsic ill-posedness, as the autocovariance operators are compact and lack bounded inverses. This paper develops a new theoretical framework for the regularized estimation and asymptotic analysis of FAR models. Leveraging Hilbert space theory, we rigorously characterize the distinction between finite- and infinite-dimensional time series analysis and formalize the necessity of regularization. To stabilize the estimation of autoregressive operators, we introduce a Tikhonov regularization scheme and derive Yule–Walker–type estimators in a general Hilbert space, and further specialize to the $L^2$ space for explicit forms. Within this unified framework, we establish the consistency and asymptotic normality of the regularized estimators and reveal that asymptotic normality can be achieved only for the predictors rather than the operator estimates themselves. Furthermore, we derive the mean squared prediction error (MSPE) and decompose its bias–variance structure. A comprehensive simulation study and an application to high-frequency functional data from wearable devices demonstrate the practical validity of the theory and the ability of FAR models to capture dynamic functional patterns.
\end{abstract}

\noindent%
\vfill

\newpage
\spacingset{1.8} 

\section{Introduction}
\label{sec: introduction}

The proliferation of data collected from continuous-time processes has brought functional data analysis (FDA) to the forefront of modern statistics (see \citealp{ramsay2005functional, hsing2015theoretical}). In many applications, functional observations often exhibit strong temporal dependence rather than independence; common examples include segments of electroencephalogram (EEG) signals \citep{hasenstab2017multi}, daily pollution concentration profiles \citep{aue2015prediction}, and daily traffic flow curves \citep{ma2024network}. To capture such dependencies, the framework of functional time series (FTS) extends classical time series analysis into infinite-dimensional spaces. A seminal work in this direction is \citet{bosq2000linear}, which formalized the functional autoregressive (FAR) model within a Hilbert space. This model generalizes the classical scalar autoregressive (AR) and vector autoregressive (VAR) models by replacing the multivariate observation vectors with random elements in a Hilbert space, and the coefficient matrices with bounded linear operators. For theoretical foundations of the FAR model, refer to \citet{HörmannSiegfried2010WDFD, horvath2014testing, kowal2019functional}. Meanwhile, surrogate approaches in \citet{aue2015prediction, chang2024modeling} leveraged functional principal component analysis (FPCA) as a dimension reduction tool to approximate the FAR model.

Despite their versatility, FAR models pose fundamental theoretical challenges that have no counterparts in finite-dimensional time series analysis. In Hilbert spaces, the autocovariance operators are compact, implying that their eigenvalues converge to zero. Consequently, their inverse operators are unbounded, rendering direct analogs of the Yule–Walker equations ill-posed. This inverse problem in operator theory \citep{kirsch2011introduction} destabilizes estimation and invalidates conventional asymptotic arguments, reflecting the intrinsic link between infinite-dimensionality and ill-posedness that constrains the development of a coherent asymptotic framework.

Most existing methods are formulated in the $L^2$ space of square-integrable functions (see \citealp{HörmannSiegfried2010WDFD, kowal2019functional}). While the $L^2$ setting is adequate for many applications, it becomes restrictive when inferential targets involve derivatives of functional trajectories \citep{mas2009functional} or when the data are inherently multivariate, such as tri-axial acceleration signals from wearable sensors in Section \ref{sec: real data}. In these cases, more general Hilbert spaces—such as Sobolev spaces or vector-valued function spaces—provide natural environments that capture smoothness and multivariate structure simultaneously. 

Furthermore, the prevailing estimation strategies rely heavily on spectral truncation \citep{mas2007weak, chen2022functional} to bypass the unbounded inverses of autocovariance operators. However, truncation introduces a delicate bias–variance trade-off: using too few eigenpairs yields large bias, while using too many amplifies variance through the inverses of small eigenvalues. As a result, existing approaches lack both numerical stability and rigorous asymptotic characterization. Besides this issue, most theoretical developments are confined to the first-order FAR(1) model \citep{antoniadis2006functional, kargin2008curve}, leaving the estimation method and asymptotic theory of high-order FAR models largely underdeveloped, which are essential for capturing complex dynamics \citep{kokoszka2013determining}.

To address these challenges, this paper develops a unified theoretical and methodological framework for regularized estimation and asymptotic analysis of functional autoregressive models in general Hilbert spaces. Rather than serving merely as an abstract extension, the Hilbert space formulation provides a rigorous foundation that clarifies the essential difference between finite- and infinite-dimensional time series analysis. In finite dimensions, autocovariance matrices are nonsingular and standard estimation techniques are well-posed. In contrast, in infinite-dimensional settings, compactness of the autocovariance operator implies intrinsic ill-posedness, necessitating regularization.

Recognizing this structural distinction, we introduce Tikhonov regularization as a principled and theoretically grounded solution. By applying a smooth penalization to counteract the instability arising from inverting small eigenvalues, Tikhonov regularization stabilizes estimation of autoregressive operators while preserving their functional structure. This regularized framework not only resolves the ill-posedness problem but also enables a tractable asymptotic theory for dependent functional data.

Our study makes three major contributions to the theory of functional time series. First, we establish a general theoretical formulation for FAR models in arbitrary separable Hilbert spaces. This framework accommodates a wide range of functional data types—including smooth curves, derivative processes, and multivariate trajectories—and elucidates why regularization is indispensable for inference in infinite-dimensional settings. The regularized inverse operator provides both theoretical tractability and numerical stability.

Second, building upon this foundation, we propose a Tikhonov-regularized Yule–Walker estimator for the autoregressive operator and establish asymptotic normality of the resulting predictor under some mild conditions. To our knowledge, this represents the first asymptotic normality result for a regularized estimator in a dependent functional setting—a significant theoretical advance bridging functional time series and inverse problem theory.

Third, we provide a bias–variance decomposition of the mean squared prediction error (MSPE), separating the effects of regularization bias, truncation bias, and estimation variance. This decomposition yields interpretable insights into the role of the regularization parameter and demonstrates the superior stability of the proposed method relative to conventional truncation-based estimators.

Complementing the theory, we conduct extensive simulation studies and apply the proposed method to multivariate functional data collected from wearable sensors. The empirical analysis confirms the theoretical properties and illustrates the model’s capacity to capture fine-grained temporal dynamics. A novel visualization scheme is further proposed to display estimated operators in three dimensions, facilitating intuitive interpretation of functional dependencies.

The remainder of the paper is organized as follows. Section~\ref{sec: preliminaries} introduces the FAR($p$) model in Hilbert space and derives the Yule–Walker equations. Section~\ref{sec: Yule-Walker estimators} presents identifiability conditions and constructs the regularized estimators, with Section~\ref{sec: example} illustrating the framework in $L^2$ spaces. Asymptotic properties, including consistency, asymptotic normality, and MSPE, are developed in Section~\ref{sec: asymptotic results}. Section~\ref{sec: simulation} reports simulation results, and Section~\ref{sec: real data} applies the proposed methodology to wearable sensor data. Section~\ref{sec: conclusion} concludes with remarks and future directions. Proofs are provided in the supplementary material.

\section{Preliminaries}
\label{sec: preliminaries}

\subsection{Notations}
\label{subsec: notations}

The notations used throughout this paper are summarized below. Let $\Z$, $\bN^+$, and $\R$ denote the sets of integers, positive integers, and real numbers, respectively. For two positive sequences $\{a_n\}, n\in\bN^+$ and $\{b_n\}, n\in\bN^+$, we write $a_n = \smallo(b_n)$ if $a_n/b_n\rightarrow 0$ as $n\rightarrow\infty$, and $a_n = \O(b_n)$ if there exists a constant $c>0$ such that $a_n\leq cb_n$ for all $n$. 

Let $\H$ and $\H^{\prime}$ be two separable Hilbert spaces with inner products $\langle \cdot, \cdot \rangle_{\H}$ and $\langle \cdot, \cdot \rangle_{\H^{\prime}}$, and the induced norms $\| \cdot \|_{\H} = \langle \cdot, \cdot \rangle_{\H}^{1/2}$ and $\| \cdot \|_{\H^{\prime}} = \langle \cdot, \cdot \rangle_{\H^{\prime}}^{1/2}$, respectively.  For simplicity, the subscripts $\H$ and $\H^{\prime}$ will be omitted unless their inclusion is necessary to avoid confusion. To streamline the discussion, we assume that all inner products in this paper yield real numbers. Suppose $\rho$ is a linear operator from $\H$ to $\H^{\prime}$ with adjoint operator $\rho^*$. The operator norm, Hilbert-Schmidt norm, and nuclear norm of $\rho$ are defined as $\| \rho \|_{\L} := \sup_{h\in\H, \|h\|_\H \leq 1} \| \rho h\|_{\H^{\prime}}$, $\| \rho \|_{\S} := \left( \sum_{i, j} \langle \rho \boldsymbol{e}_i, \boldsymbol{e}_j^{\prime} \rangle^2_{\H^{\prime}} \right)^{1/2}$, and $\| \rho \|_{\N} := \sum_i \langle (\rho^* \rho)^{1/2} \boldsymbol{e}_i, \boldsymbol{e}_i\rangle_{\H}$, respectively, for some complete orthonormal systems $\{\boldsymbol{e}_i\}, i\in\bN^+$ of $\H$ and $\{\boldsymbol{e}_j^{\prime}\}, j\in\bN^+$ of $\H^{\prime}$. The operator $\rho$ is said to be bounded if its operator norm is finite. Denote $0_{\H}$, $\mathscr{O}_{\H}$ and $\I_{\H}$ as \enquote{zero} in $\H$, the zero operator and the identity operator on $\H$, respectively.

Suppose that the random elements are defined on a common probability space $(\Omega, \F, \P)$. The notations $\overset{\P}{\rightarrow}$ and $\overset{d}{\rightarrow}$ represent \emph{convergence in probability} and \emph{convergence in distribution}, respectively. An $\H$-valued random variable is a measurable function from $\Omega$ to $\H$, where $\H$ is a Hilbert space and equipped with its Borel $\sigma$-algebra. 

Let $\{Z_n\}, n\in\Z$ be a sequence of $\H$-valued random elements and $\{u_n\}, n\in\Z$ be a sequence of real-valued random variables. We write $Z_n=\smallo_{\P}(u_n)$ if $Z_n/u_n\overset{\P}{\rightarrow}0_{\H}$. Denote $\H^p \ (p \geq 2)$ as the $p$-fold Cartesian product of $\H$. $\H^p$ is thereby a separable Hilbert space with the inner product $\langle \bh, \widetilde{\bh} \rangle = \sum_{i=1}^p \langle h_i, \widetilde{h}_i \rangle$ for $\bh = (h_1, \ldots, h_p)^\T$ and $\widetilde{\bh} = (\widetilde{h}_1, \ldots, \widetilde{h}_p)^\T$ in $\H^p$. For $h\in\H$ and $h^{\prime}\in\H^{\prime}$, the
notation $h \otimes h^{\prime}$ is defined as a linear and bounded operator from $\H$ to $\H^{\prime}$ such that $h \otimes h^{\prime} = \langle h,\cdot \rangle_{\H} h^{\prime}$.

\subsection{Functional Autoregressive Model in Hilbert Space}
\label{subsec: ARH model}

Consider a stationary process $\boldsymbol{X}=\left(X_t \right), t \in \Z$ with each $X_t \in \H$. Assume that $\E[X_t]=0_{\H}$ for all $t \in \Z$, where the expectation is understood in the Bochner-integral sense \citep{hsing2015theoretical}. The \emph{functional autoregressive model of order $p$ (FAR($p$)) in Hilbert space} is defined by (see \citealt{bosq2000linear}),
\begin{equation}
    \label{eq: ARH}
    X_t=\rho_1 X_{t-1} + \cdots + \rho_p X_{t-p} +\varepsilon_t, \ t \in \Z,
\end{equation}
where $\rho_1, \ldots, \rho_p$ are linear and bounded operators on $\H$ with $\rho_p \neq \mathscr{O}_{\H}$, referred to as \emph{autoregressive operators}, and $\boldsymbol{\varepsilon}=\left(\varepsilon_t\right), t \in \Z$ is a zero-mean, independent and identically distributed (i.i.d.) innovation sequence in $\H$.

The FAR($p$) model extends the classical autoregressive (AR) framework \citep{brockwell1991time} and its vector-valued counterpart (VAR) \citep{stock2001vector} to an infinite-dimensional Hilbert space. This extension provides a flexible and rigorous framework for modeling and analyzing stochastic processes whose realizations are functions.

By incorporating higher-order dependencies, the FAR($p$) model enhances flexibility of FAR(1) to capture the full dynamics in functional time series. However, unlike the inherently Markovian FAR(1), the FAR($p$) model requires a transformation to achieve Markov properties, known as the \emph{Markov representation}. This transformation is achieved by constructing $Y_t := (X_t, \ldots, X_{t-p+1})^\T \in \H^p$, then the resulting process $\boldsymbol{Y} := (Y_t), t \in \Z$ is stationary, and the orignal FAR($p$) process $\boldsymbol{X}$ in the Hilbert space $\H$ is equivalent to the FAR(1) process $\boldsymbol{Y}$ in the product space $\H^p$ by
\begin{align}
    \label{eq: ARHp}
    \text{Equation }(\ref{eq: ARH}) &\Longleftrightarrow Y_t = \bPhi Y_{t-1} +  \zeta_t  \\
    \label{eq: ARHp_matrix}
        :&\Longleftrightarrow
    \begin{pmatrix}
        X_t\\
        X_{t-1}\\
        \vdots\\
        X_{t-p+1}
    \end{pmatrix} =
    \begin{pmatrix}
        \rho_1 & \cdots & \cdots & \cdots & \rho_p\\
        \I_{\H} & \mathscr{O}_{\H} & \cdots & \cdots & \mathscr{O}_{\H} \\
        \vdots & \ddots & \ddots & \ddots & \vdots\\
        \mathscr{O}_{\H} & \cdots & \mathscr{O}_{\H} & \I_{\H} & \mathscr{O}_{\H}
    \end{pmatrix}
    \begin{pmatrix}
        X_{t-1}\\
        X_{t-2}\\
        \vdots\\
        X_{t-p}
    \end{pmatrix} +
    \begin{pmatrix}
        \varepsilon_{t}\\
        0_{\H}\\
        \vdots\\
        0_{\H}
    \end{pmatrix},
\end{align}
where $\bPhi$ is a linear and bounded operator on $\H^p$ and $\zeta_t = (\varepsilon_{t}, 0_\H, \ldots, 0_\H)^\T\in \H^p$ represents the augmented innovation sequence. 

\begin{remark}
    Equation~(\ref{eq: ARHp_matrix}) specifies $\bPhi$ as an operator-valued matrix acting on $\H^p$, where the operation mimics standard matrix multiplication with scalar multiplication replaced by operator action. We adopt this notation hereafter without elaboration.
\end{remark}

\subsection{The Yule-Walker Equations for FAR($p$) Model}
\label{subsec: Yule-Walker equations}

The identification and prediction of the FAR($p$) model~(\ref{eq: ARH}) rely on estimating both the autoregressive operators $\rho_1, \ldots, \rho_p$ and the autocovariance operator of the innovation process $\boldsymbol{\varepsilon}$. To this end, we employ the Yule–Walker equations \citep{brockwell1991time, dou2016generalized, zadrozny2016extended}, which utilize the second-order moment information to provide an estimation procedure. 

We derive the Yule–Walker equations through the Markov representation~(\ref{eq: ARHp_matrix}) of the FAR($p$) model in $\H^p$, leveraging properties of operators induced operation $\otimes$; see Appendix A in the supplementary material for details. Pre-multiplying both sides of Equation~(\ref{eq: ARHp}) by $Y_k, k=t,t-1$ in $\otimes$ sense and subsequently taking expectations, we obtain
\begin{equation}
    \label{eq: expectation of tensor}
    \E [Y_k \otimes Y_t] = \bPhi \E [Y_k \otimes Y_{t-1}] + \E [Y_k \otimes \zeta_t], \ k=t,t-1.
\end{equation}

\begin{remark}
    The expectation $\E [Y_k \otimes Y_t]$ in Equation~(\ref{eq: expectation of tensor}) is defined as a bounded linear operator on $\H^p$ satisfying $(\E [Y_k \otimes Y_t])\bh = \E [\langle Y_k, \bh \rangle Y_t]$ for all $\bh \in \H^p$, where $\E [\langle Y_k, \bh \rangle Y_t]$ denotes the Bochner expectation of the $\H^p$-valued random variable $\langle Y_k, \bh \rangle Y_t$.
\end{remark}

To relate the operator $\bPhi$ on $\H^p$, defined in Equation~(\ref{eq: ARHp_matrix}), to the autoregressive operators $\rho_1, \ldots, \rho_p$ on $\H$, we introduce projection operators $\pi_j: \H^p \to \H$ with $\pi_j(\bh) = h_j$ for $j = 1, \ldots, p$, along with their adjoint operators $\pi_j^*: \H \to \H^p$ given by $\pi_j^* (h_j) = (0_{\H}, \ldots, h_j, \ldots, 0_{\H})^\T$, where $h_j$ occupies the $j$-th coordinate. We then pre-multiply both sides of Equation~(\ref{eq: expectation of tensor}) by $\pi_1$ for $k = t-1, t$, and further post-multiply by $\pi_1^*$ when $k=t$, yielding 
\begin{align*}
    \E [Y_{t-1} \otimes X_t] &= (\rho_1, \ldots, \rho_p) \E[Y_{t-1} \otimes Y_{t-1}], \\
    \E [X_t \otimes X_t] &= (\rho_1, \ldots, \rho_p) \E[X_t \otimes Y_{t-1}] + \E[X_t \otimes \varepsilon_t].
\end{align*}
Let $\brho := (\rho_1, \ldots, \rho_p)$ denote the unknown parameter, $\D := \E[Y_t \otimes Y_t]$ the autocovariance operator of $\boldsymbol{Y}$, $\cE := \E [Y_t \otimes X_{t+1}]$ the lag-1 cross-covariance operator of $\boldsymbol{Y}$ and $\boldsymbol{X}$, $\C_h := \E[X_0 \otimes X_h]$ the lag-$h$ autocovariance operator of $\boldsymbol{X}$, and $\C_{\boldsymbol{\varepsilon}} := \E[\varepsilon_0 \otimes \varepsilon_0]$ the autocovariance operator of $\boldsymbol{\varepsilon}$. Due to the stationarity, $\D$ and $\cE$ are independent of $t$. With these definitions, the Yule-Walker equations can be expressed in the following concise forms:
\begin{align}
    \label{eq: Y-W-1}
    \cE &= \brho\D,\\
    \label{eq: Y-W-2}
    \C_0 &= \brho\cE^* + \C_{\boldsymbol{\varepsilon}}.
\end{align}

The operator $\D$ and $\cE$ exhibit matrix representations with operator-valued entries:
\begin{align}
    \label{eq: D_E form}
    \D = 
    \begin{pmatrix}
        \C_0 & \C_1 & \cdots & \C_{p-1}\\
        \C_1^* & \C_0  & \cdots & C_{p-2} \\
        \vdots & \vdots & \ddots & \vdots\\
        \C_{p-1}^* & \C_{p-2}^* & \cdots & \C_0
    \end{pmatrix}, \quad
    \cE = (\C_1, \ldots, \C_p).
\end{align}
This structured representation offers computational advantages, as the repeating pattern of autocovariance operators $\C_0, \ldots, \C_{p-1}$ in $\D$ and $\cE$ simplifies numerical implementation. The Yule-Walker equations reveal the linear relationships among these autocovariance operators up to the $p$-th lag. Derived from the second-order moments, these equations serve as a cornerstone for estimating the parameters of the FAR($p$) model.

\section{The Yule-Walker Estimation}
\label{sec: Yule-Walker estimators}

\subsection{Identifiability and Regularization}
\label{subsec: identifiability}

For statistical inference and asymptotic theory based on Equation~(\ref{eq: Y-W-1}), the eigenanalysis of the autocovariance operator $\D$ is essential. 
Under the condition $\E \|Y\|^2 < \infty$, the autocovariance operator $\D$ is self-adjoint, positive semi-definite, and compact on $\H^p$.
Therefore, $\D$ admits the spectral decomposition \citep{hsing2015theoretical}
\begin{equation}
    \label{eq: spectral decomposition}
    \D = \sum_{j=1}^{\infty} \lambda_j \bnu_j \otimes \bnu_j,
\end{equation}
where $\lambda_j, j \in \bN^+$ is a non-increasing sequence of eigenvalues and $\bnu_j, j \in \bN^+$ are the corresponding orthonormal eigenvectors in $\H^p$. 

Identifiability is a fundamental concept in parameter estimation, ensuring that distinct parameter values correspond to distinct model structures. We first characterize when the autoregressive operator $\brho$ is uniquely determined throughout the entire space $\H^p$. The proofs and supplementary results of this section can be found in Appendix B.

\begin{proposition}
    \label{prop: identifiability on Hp}
    The following statements are equivalent: \\
    (i) (Global identifiability) If $\brho\D = \widetilde{\brho}\D$ for some linear and bounded operator $\widetilde{\brho}$, then $\brho = \widetilde{\brho}$. \\
    (ii) (Injectivity) The null space $\ker\D:=\{\by \in \H^p: \D \by=0_{\H^p}\} = \{0_{\H^p}\}$. \\
    (iii) (Positive eigenvalues) All eigenvalues $\lambda_j$ of $\D$ are positive.
\end{proposition}

Proposition~\ref{prop: identifiability on Hp}(ii) indicates that global identifiability on $\H^p$ requires the image space ${\rm Im} \D := \{\D \by: \by \in \H^p\}$ to be sufficiently rich to distinguish between different $\brho$. However, identifiability over the entire space $\H^p$ is not necessary. The key insight is that $\brho$ only needs to be identified on ${\rm Im} \D$, as constrained by Equation~(\ref{eq: Y-W-1}).

\begin{proposition}[Identifiability on $(\ker\D)^{\perp}$]
    \label{prop: identifiability on orthocomplement}
    If $\brho\D = \widetilde{\brho}\D$ for some linear and bounded operator $\widetilde{\brho}$, then $\brho|_{(\ker\D)^{\perp}} = \widetilde{\brho}|_{(\ker\D)^{\perp}}$.
\end{proposition}

Proposition~\ref{prop: identifiability on orthocomplement} demonstrates that $\brho$ is identifiable on $(\ker\D)^{\perp}$, regardless of whether $\D$ is injective. The space $(\ker\D)^{\perp}$ coincides with $\overline{\text{span}} \{\bnu_j: \lambda_j>0\}$, the closure of the span of all eigenfunctions corresponding to positive eigenvalues. Since ${\rm Im} \D$ is dense in $(\ker\D)^{\perp}$, identifiability on ${\rm Im} \D$ follows immediately, without any additional assumptions.

\begin{remark}
     Identifiability in infinite-dimensional models contrasts sharply with the finite-dimensional case. In $d$-dimensional VAR($p$) models, autoregressive matrices are identifiable because the autocovariance matrix is invertible. Here, letting $d_0 = \sup \{j\in\bN^+: \lambda_j>0\}$, Lemma B.1 shows that ${\rm Im} \D = (\ker\D)^{\perp} \subsetneqq \H^p$ when $d_0 < \infty$, and ${\rm Im} \D = \left\{\bz \in \H^p: \sum_{j=1}^{\infty} \frac{\langle \bz, \bnu_j \rangle^2}{\lambda_j^2} < \infty\right\} \subsetneqq \H^p$ when $d_0 = \infty$. Thus, $\D$ is never surjective, and its compactness imposes inherent identifiability limitations.
%
\end{remark}

We provide a sufficient condition, which refers to the Picard condition \citep{hsing2015theoretical}, for the existence of the solution to Equation (\ref{eq: Y-W-1}) on the identifiable space $(\ker\D)^{\perp}$.

\begin{proposition}
    \label{prop: existence}
    If the condition $\sum_{j=1}^{d_0} \lambda_j^{-2}\langle x, \cE \bnu_j \rangle^2 < \infty$ holds for all $x \in \H$, then the solution for $\brho$ to $\cE = \brho\D$ restricted on $(\ker\D)^{\perp}$ exists and is unique, with its adjoint operator given by $(\brho|_{(\ker\D)^{\perp}})^* x = \sum_{j=1}^{d_0} \lambda_j^{-1}\langle x, \cE \bnu_j \rangle \bnu_j$ for all $x \in \H$.
\end{proposition}

Although identifiability is established, directly deducing $\brho = \cE\D^{-1}$ by inverting $\D$ as in the finite-dimensional case is not feasible. The issue stems from the spectral property of compact operators in infinite-dimensional Hilbert spaces: the eigenvalues of $\D$ form at most a countable sequence converging to zero, with zero as the only accumulation point (Theorem 4.2.3 in \citealt{hsing2015theoretical}). Consequently, its generalized inverse $\D^{-1}=\sum_{j=1}^{\infty} \lambda_j^{-1} \bnu_j \otimes \bnu_j$ is an unbounded operator and is highly sensitive to small perturbations in the data. 

To overcome this ill-posedness, we introduce regularization techniques to stabilize the inversion of $\D$. Traditional approaches often rely merely on truncation-based approximations of the autocovariance operator \citep{mas2007weak}, but these methods are highly sensitive to the choice of truncation dimension and may lead to unstable inversion. Instead, we propose the Tikhonov-truncation regularization approach to enhance numerical stability.

Specifically, let $k_n, n \in \bN^+$ and $\alpha_n, n \in \bN^+$ be two sequences of positive numbers such that $k_n \to \infty$ and $\alpha_n \to 0$ as $n \to \infty$, representing the truncation level and the regularization parameter, respectively. The stabilized inverse operator $\Dd$ is defined as
\begin{equation}
    \Dd := \sum_{j=1}^{k_n} \frac{\lambda_j}{\lambda_j^2 + \alpha_n} \bnu_j \otimes \bnu_j.
\end{equation}
This regularization modifies the eigenvalues of the inverse from $1/\lambda_j$ to $\lambda_j/(\lambda_j^2 + \alpha_n)$, effectively controlling the impact of small eigenvalues and ensuring a well-posed approximation. 

\begin{remark}
    Our method combines the strengths of truncation and Tikhonov regularization. Compared with pure truncation, it applies smooth spectral shrinkage to small eigenvalues, reducing sensitivity to the cutoff $k_n$ and retaining more autocovariance information. Compared with pure Tikhonov regularization \citep{kirsch2011introduction}, it achieves computational efficiency through finite-dimensional approximation while remaining theoretically consistent as $k_n \to \infty$. Theorem~\ref{thm: MSE} and simulations in Section~\ref{sec: simulation} demonstrate its superior bias–variance trade-off and overall performance.
\end{remark}

\subsection{The Yule-Walker Estimators and Predictor}
\label{subsec: Yule-Walker estimators}

Suppose that the samples $X_1,\ldots,X_n$ are observed. The empirical versions of the lagging autocovariance operator $\C_h$, the autocovariance operator $\D$, and the cross-covariance operator $\cE$ are respectively estimated as follows:
\begin{align}
    \label{eq: empirical C}
    \widehat{\C}_h &=
    \left\{
    \begin{array}{llr}
        \frac{1}{n-|h|} \sum\limits_{t=|h|+1}^n X_t\otimes X_{t+h}, & -n<h<0, \\
        \frac{1}{n- h } \sum\limits_{t=1}^{n-h}X_t \otimes X_{t+h}, & 0 \leq h<n,
    \end{array}\right. \\
    \hD &= \frac{1}{n-p} \sum_{t=p+1}^n Y_t \otimes Y_t, \notag \\
    \hE &= \frac{1}{n-p} \sum_{t=p+1}^n Y_{t-1} \otimes X_t. \notag
\end{align}
Let $(\widehat{\lambda}_j, \hnu_j), j\in\bN^+$ be the eigenpairs of $\hD$ corresponding to the computation of the functional principal component analysis (FPCA) \citep{james2000principal, yao2005functional, zhou2022theory} of $\hD$. Denote by $\hD^{\dagger}$ the stabilized inverse operator of $\hD$. The estimator $\hrho$ for $\brho$ is derived based on the Yule-Walker equation~(\ref{eq: Y-W-1}) as follows:
\begin{equation}
    \hrho := \hE \hD^{\dagger} = \hE \sum_{j=1}^{k_n} \frac{\widehat{\lambda}_j}{\widehat{\lambda}_j^2 + \alpha_n} \hnu_j \otimes \hnu_j.
\end{equation}
After observing new samples $Y_{n+p} = (X_{n+p}, \ldots, X_{n+1})^\T$, the subsequent unobserved $X_{n+p+1}$ can be predicted as
\begin{equation}
    \widehat{X}_{n+p+1} = \hrho Y_{n+p}.
\end{equation}
For the estimation of the autocovariance operator $\C_{\boldsymbol{\varepsilon}}$ and variance $\sigma_{\boldsymbol{\varepsilon}}^2 := \E \| \varepsilon_1 \|^2$ of innovations, we adopt the plug-in estimators based on Equation~(\ref{eq: Y-W-2}),
\begin{equation}
    \hC_{\boldsymbol{\varepsilon}} := \hC_0 - \hrho \hE^* \ \text{ and } \  \widehat{\sigma}_{\boldsymbol{\varepsilon}}^2 := \|\hC_{\boldsymbol{\varepsilon}}\|_{\N}.
\end{equation}

\subsection{Example of Yule-Walker Estimators in $\H=L^2([0,1], \R^d)$}
\label{sec: example}

While the Yule-Walker estimation framework has been established for general Hilbert spaces, practical implementation requires handling discrete observations of underlying random functions. To address this, we focus on the Hilbert space $\H = L^2([0,1], \R^d)$, which consists of square-integrable functions from $[0,1]$ to $\R^d \ (d\in\bN^+)$. This canonical space accommodates both univariate and multivariate functional data (e.g., curves, trajectories, multi-channel signals), making it particularly relevant for methodological development and applications. We restrict our analysis to settings where $d$ is considerably smaller than the sample size, excluding high-dimensional cases, and investigate FAR($p$) models with particular attention to estimating the autocovariance and autoregressive operators.

Suppose we have $n$ samples $\boldsymbol{X}_1(u),\ldots,\boldsymbol{X}_n(u)\in L^2([0,1], \R^d)$ that follow the FAR($p$) model (\ref{eq: ARH}) with autoregressive operators $\rho_1,\ldots,\rho_p$ and white noises $\boldsymbol{\varepsilon}_1(u),\ldots,\boldsymbol{\varepsilon}_n(u)$. Here $u \in [0, 1]$, both $\boldsymbol{X}_t(u)=(X_{t1}(u),\cdots,X_{td}(u))^\T$ and $\boldsymbol{\varepsilon}_t(u)=(\varepsilon_{t1}(u),\ldots,\varepsilon_{td}(u))^\T, \ t=1,\ldots,n$ are $d$-dimensional. The FAR($p$) model is concretized as
\begin{align}
    \label{eq: special ARHp}
    \begin{pmatrix} X_{t1}(u) \\ \vdots \\ X_{td}(u) \end{pmatrix}
    =\begin{pmatrix}
        \rho_1^{11} & \cdots & \rho_1^{1d}\\
        \vdots & \ddots & \vdots\\
        \rho_1^{d1} & \cdots &  \rho_1^{dd}
    \end{pmatrix}
    \begin{pmatrix} X_{t-1,1}(u) \\ \vdots \\ X_{t-1,d}(u) \end{pmatrix}
    +\cdots+
    \begin{pmatrix}
        \rho_p^{11} & \cdots & \rho_p^{1d}\\
        \vdots & \ddots & \vdots\\
        \rho_p^{d1} & \cdots &  \rho_p^{dd}
    \end{pmatrix}
    \begin{pmatrix} X_{t-p,1}(u) \\ \vdots \\ X_{t-p,d}(u) \end{pmatrix}
    +\begin{pmatrix} \varepsilon_{t1}(u) \\ \vdots \\ \varepsilon_{td}(u) \end{pmatrix},
\end{align}
where each $\rho_j^{kl}, j=1,\ldots,p; k,l=1,\ldots,d$, the $(k,l)$-entry of $\rho_j$, is a linear and bounded operator on $L^2([0,1], \R)$. Consider the case when each $\rho_j^{kl}$ is an integral operator, that is, there exists some square-integrable function $R_j^{kl}(u, v)$ on $[0,1]\times [0,1]$ such that
\begin{equation}
    \label{eq: kernel operator}
    \rho_j^{kl} (X(v)) (u) = \int_0^1 R_j^{kl}(u,v) X(v) d v,
\end{equation}
for any $X\in L^2([0,1], \R)$. $R_j^{kl}(u,v)$ is referred to as the \emph{kernel} of $\rho_j^{kl}$, which captures the interaction between the $k$-th and $l$-th components of the functional time series for lag $j$. In this way, an abstract operator corresponds to a bivariate function. Of interest is to estimate all $R_j^{kl}(u,v)$ and utilize them to predict the subsequent unobserved samples.

\noindent\textbf{Continuous version of the Yule-Walker equations}

It is more convenient to apply the Yule-Walker Equations~(\ref{eq: Y-W-1}), (\ref{eq: Y-W-2}) and~(\ref{eq: D_E form}) for calculations, whose forms transform from operator-based equations to those involving their corresponding kernels. Specifically, the equations are given by
\begin{align}
    &(C_1(u,v), \ldots, C_p(u,v)) \notag \\
    \label{eq: Y-W-cont-1}
    = & \int_0^1 (R_1(u,z), \ldots, R_p(u,z))
    \begin{pmatrix}
        C_0(z,v) & C_1(z,v) & \cdots & C_{p-1}(z,v)\\
        C_1(v,z) & C_0(z,v)  & \cdots & C_{p-2}(z,v) \\
        \vdots & \vdots & \ddots & \vdots\\
        C_{p-1}(v,z) & C_{p-2}(v,z)& \cdots & C_0(z,v)
    \end{pmatrix} d z,\\
    \label{eq: Y-W-cont-2}
    C_{\boldsymbol{\varepsilon}}(u,v) &= C_0(u,v) - \int_0^1 (R_1(u,z), \ldots, R_p(u,z))
    \begin{pmatrix}
        C_1(v,z) \\
        \vdots \\
        C_p(v,z)
    \end{pmatrix} d z,
\end{align}
where $C_h(u,v)$, $C_{\boldsymbol{\varepsilon}}(u,v)$ and $R_j(u,v)$ are the kernels of $\C_h$, $\C_{\boldsymbol{\varepsilon}}$ and $\rho_j$, respectively, for $h=0,1,\ldots,p$ and $j=1,\ldots,p$. The estimation of autoregressive operators lies in determining the kernels of the empirical autocovariance operators $\widehat{\C}_h$. Based on Equation~(\ref{eq: empirical C}), after centralizing the samples, the kernel of $\widehat{\C}_h$ can be expressed as 
\begin{align*}
    \widehat{C}_h (u,v) &= \frac{1}{n-h} \sum_{t=1}^{n-h} \boldsymbol{X}_{t+h}(u) \boldsymbol{X}_{t}^\T(v) \\
    &= \frac{1}{n-h} \sum_{t=1}^{n-h}
    \begin{pmatrix}
        X_{t+h,1}(u) X_{t1}(v) & \cdots & X_{t+h,1}(u) X_{td}(v) \\
        \vdots & \ddots & \vdots\\
        X_{t+h,d}(u) X_{t1}(v) & \cdots & X_{t+h,d}(u) X_{td}(v)
    \end{pmatrix}
\end{align*}
for $h \geq 0$. Substituting $C_h(u,v)$ in Equation~(\ref{eq: Y-W-cont-1}) and (\ref{eq: Y-W-cont-2}) with their empirical kernels $\widehat{C}_h (u,v)$ and employing iterative methods \citep{WazwazAbdul-Majid2011LaNI}, we can derive the Yule-Walker estimations $\widehat{R}_j(u,v)$ and $\widehat{C}_{\boldsymbol{\varepsilon}}(u,v)$ for $R_j(u,v)$ and $ C_{\boldsymbol{\varepsilon}}(u,v)$, respectively. Then the prediction can be constructed by
\begin{equation}
    \widehat{\boldsymbol{X}}_{n+p+1} (u) = \int_0^1 \left(\widehat{R}_1(u,v)\boldsymbol{X}_{n+p} (v) + \cdots + \widehat{R}_p(u,v)\boldsymbol{X}_{n+1} (v) \right) d v.
\end{equation}

\noindent\textbf{Discrete version of the Yule-Walker equations}

In practice, only discrete observations of each function are available. Assume that for each $X_{tr}(u)$, we observe $\widetilde{X}_{tr}=(X_{tr}(s_1),\ldots,X_{tr}(s_g))^\T$, where $s_1, \ldots, s_g$ are $g$ equidistant points in $[0, 1]$ with $0=s_1<s_2<\cdots<s_g=1$ for $t=1,\ldots,n$ and $r=1.\ldots,d$. Unequally spaced observations can be preprocessed via interpolation. The concatenated vector $\widetilde{\boldsymbol{X}}_t = (\widetilde{X}_{t1}^\top,\ldots,\widetilde{X}_{td}^\top)^\top$ represents the stacked discrete observations across all dimensions. From the perspective of approximating integrals by summations, the discrete version of the Yule-Walker equations becomes
\begin{align}
    \label{eq: Y-W-disc-1}
    (\widetilde{\boldsymbol{C}}_1, \ldots, \widetilde{\boldsymbol{C}}_p) =&~ (\widetilde{\boldsymbol{R}}_1, \ldots, \widetilde{\boldsymbol{R}}_p)
    \begin{pmatrix}
        \widetilde{\boldsymbol{C}}_0 & \widetilde{\boldsymbol{C}}_1 & \cdots &\widetilde{\boldsymbol{C}}_{p-1}\\
        \widetilde{\boldsymbol{C}}_1^\T & \widetilde{\boldsymbol{C}}_0  & \cdots & \widetilde{\boldsymbol{C}}_{p-2} \\
        \vdots & \vdots & \ddots & \vdots\\
        \widetilde{\boldsymbol{C}}_{p-1}^\T & \widetilde{\boldsymbol{C}}_{p-2}^\T& \cdots & \widetilde{\boldsymbol{C}}_0
    \end{pmatrix},\\
    \label{eq: Y-W-disc-2}
    \widetilde{\boldsymbol{C}}_{\boldsymbol{\varepsilon}} =&~ \widetilde{\boldsymbol{C}}_0 - (\widetilde{\boldsymbol{R}}_1, \ldots, \widetilde{\boldsymbol{R}}_p)
    \begin{pmatrix}
        \widetilde{\boldsymbol{C}}_1^\T \\
        \vdots \\
        \widetilde{\boldsymbol{C}}_p^\T
    \end{pmatrix},
\end{align}
where $\widetilde{\boldsymbol{C}}_h$, $\widetilde{\boldsymbol{C}}_{\boldsymbol{\varepsilon}}$, and $\widetilde{\boldsymbol{R}}_j$ are $dg\times dg$ discretized matrices of the kernels $C_h(u,v)$, $C_{\boldsymbol{\varepsilon}}(u,v)$, and $R_j(u,v)$, respectively. Each matrix is constructed by evaluating the corresponding kernel function on a $g \times g$ uniform grid over $[0,1]^2$ and scaling the entries by $1/g$ to approximate the integral operator. Specifically, for $h = 0, 1, \ldots, p$,
\begin{equation*}
    \widetilde{\boldsymbol{C}}_h = \frac{1}{g} \frac{1}{n-h} \sum_{t=1}^{n-h} \widetilde{\boldsymbol{X}}_{t+h} \widetilde{\boldsymbol{X}}_{t}^\T.
\end{equation*}

In the discrete formulation (\ref{eq: Y-W-disc-1}), the inverse of a $pdg\times pdg$ matrix is required. By applying the proposed regularization approach in Section \ref{subsec: identifiability} with parameters $k_n$ and $\alpha_n$, we obtain the estimations $\widehat{\widetilde{\boldsymbol{R}}}_j$ and $\widehat{\widetilde{\boldsymbol{C}}}_{\boldsymbol{\varepsilon}}$. Multiplying these estimations by $g$ provides discrete versions of $\widehat{R}_j(u,v)$ and $\widehat{C}_{\boldsymbol{\varepsilon}}(u,v)$.

Then the prediction can be proposed by
\begin{equation}
    \widehat{\widetilde{\boldsymbol{X}}}_{n+p+1} = \widehat{\widetilde{\boldsymbol{R}}}_1 \widetilde{\boldsymbol{X}}_{n+p} + \cdots + \widehat{\widetilde{\boldsymbol{R}}}_p \widetilde{\boldsymbol{X}}_{n+1}.
\end{equation}
According to Equation~(\ref{eq: CI}) below, the $(1-\alpha)$ confidence interval for the true value of $\widetilde{\boldsymbol{R}}_1 \widetilde{\boldsymbol{X}}_{n+p} + \cdots + \widetilde{\boldsymbol{R}}_p \widetilde{\boldsymbol{X}}_{n+1}$ is given point-wise by
\begin{equation}
    \left[ \widehat{\widetilde{\boldsymbol{X}}}_{n+p+1} \pm z_{1-\alpha/2} \sqrt{\frac{\widehat{t}_n {\rm diag}( \widetilde{\boldsymbol{C}}_{\boldsymbol{\varepsilon}})}{n}} \right],
\end{equation}
where ${\rm diag} (\widetilde{\boldsymbol{C}}_{\boldsymbol{\varepsilon}})$ is a vector formed by the diagonal elements of the matrix $\widetilde{\boldsymbol{C}}_{\boldsymbol{\varepsilon}}$, $\widehat{t}_n$ is as shown in Equation~(\ref{eq: t_n}) with $\lambda_j$ replaced by $\widehat{\lambda}_j$, and $z_{1-\alpha/2}$ denotes the $(1-\alpha)$-th quantile of the standard normal distribution.

\section{Asymptotic Results}
\label{sec: asymptotic results}

\subsection{Model Assumptions}
\label{subsec: model assumptions}
In this section, we require some assumptions to establish the asymptotic properties of the proposed estimator $\hrho$ and predictor $\widehat{X}_{n+p+1}$.

\begin{assumption}
    \label{ass: moment condition}
    Let $\boldsymbol{\varepsilon}$ in the model (\ref{eq: ARH}) be a strong white noise, and $\E \|X_1\|^4 < \infty$.
\end{assumption}

\begin{assumption}
    \label{ass: operator condition}
     There exists an integer $j_0 \geq 1$ such that $\|\bPhi^{j_0}\|_{\L} < 1$.
\end{assumption}

\begin{assumption}
    \label{ass: identifiability condition}
    (i) All eigenvalues of $\D$ are positive and distinct. (ii) The Picard condition $\sum_{j=1}^{\infty} \lambda_j^{-2}\langle x, \cE \bnu_j \rangle^2 < \infty, \forall x \in \H$ holds.
\end{assumption}

\begin{assumption}
    \label{ass: convexity of the function of eigenvalues}
    For $j \geq 2$, the eigenvalues satisfy $\lambda_j - \lambda_{j+1} \leq \lambda_{j-1} - \lambda_j$. 
\end{assumption}

The four assumptions above serve as regularity conditions. The finite fourth-moment condition in Assumption \ref{ass: moment condition} is necessary for establishing the asymptotic normality of the predictor. Assumption \ref{ass: operator condition} guarantees the existence of a unique stationary solution to model (\ref{eq: ARH}). Assumption~\ref{ass: identifiability condition} guarantees identifiability as shown in Propositions~\ref{prop: identifiability on Hp}–\ref{prop: existence}, and also circumvents potential complications arising from eigen-subspaces of dimension greater than one.

Assumption \ref{ass: convexity of the function of eigenvalues} imposes a mild restriction on the decay rate of the eigenvalues of $\D$, as it holds in many classical cases, including arithmetic decay ($\lambda_j = C/j^{1+\alpha}, \alpha >0$) or exponential decay ($\lambda_j = C \exp (-\alpha j), \alpha > 0$), where $C$ is a positive constant \citep{mas2007weak}. Notably, even if the inequality $\lambda_j - \lambda_{j+1} \leq \lambda_{j-1} - \lambda_j$ fails to hold for finitely many indices $j$, the asymptotic properties established in this section remain valid. For simplicity, we continue to adopt Assumption \ref{ass: convexity of the function of eigenvalues}. We define the eigenvalue gap as $\delta_j := \lambda_j-\lambda_{j+1}$ for $j\geq 1$.

Denote $\Pi_k$ and $\widehat{\Pi}_k$ as the projectors onto the space spanned by the first $k$ eigenvectors of $\D$ and $\hD$, respectively. By Karhunen–Lo{\`e}ve (K-L) expansion, $Y_t$ can be represented as
\begin{equation}
    \label{eq: KL}
    Y_t \overset{d}{=} \sum_{j=1}^{\infty} \eta_{j}^{(t)}\sqrt{\lambda_j}\boldsymbol{\nu}_j
\end{equation}
under Assumptions \ref{ass: moment condition}-\ref{ass: operator condition}, where \enquote{$\overset{d}{=}$} denotes equality in distribution, and $\eta_{j}^{(t)}$ are uncorrelated real-valued random variables with zero mean and unit variance.

\subsection{Consistency}
\label{subsec: consistency}
We establish the following consistency properties for the Yule-Walker estimator $\hrho$.

\begin{theorem}
    \label{thm: consistency}
    Let Assumption \ref{ass: moment condition}-\ref{ass: convexity of the function of eigenvalues} hold.  \\
    (i) If $k_n \equiv K$ is finite, and $\alpha_n = \smallo \left(n^{-1/4}(\log n)^{\beta}\right)$ for some $\beta > 1/2$, then
    \begin{equation}
        \|\hrho - \brho \Pi_K\|_{\S} = \smallo_{\P} \left(\frac{(\log n)^{\beta}}{n^{1/4}}\right).
    \end{equation}
    (ii) If $k_n\rightarrow\infty$ as $n\rightarrow\infty$, $\alpha_n \lambda_{k_n}^{-2} = \smallo (1)$ and $k_n \gamma_{k_n}/\lambda_{k_n} = \smallo \left(n^{1/4}(\log n)^{-\beta}\right)$ for some $\beta > 1/2$, where $\gamma_j := (\lambda_j - \lambda_{j+1})^{-1}$, then
    \begin{equation}
        \|\hrho - \brho\Pi_{k_n} \|_{\S} = \smallo_{\P}(1).
    \end{equation}
    (iii) Under the same conditions as in (ii), if $\sum_{j=k_n+1}^{\infty} \|\brho(\bnu_j)\|^2 = \smallo(1)$, then
    \begin{equation}
        \|\hrho - \brho\|_{\S} = \smallo_{\P}(1).
    \end{equation}
\end{theorem}

The results of Theorem \ref{thm: consistency} are derived through a step-by-step analysis based on the projection of the true operator $\brho$. Ideally, one would directly evaluate $\hrho - \brho$. However, since the construction of $\hrho$ relies solely on the first $k_n$ eigenvalues of $\hD$, it is natural to project $\brho$ onto the space spanned by the first $k_n$ eigenvectors of $\D$. If, in addition, the contribution of $\brho$ along the directions associated with the tail eigenvectors is sufficiently small, the consistency of $\hrho$ with $\brho$ can be established.

\subsection{Asymptotic Normality}
\label{subsec: CLT}

To further establish the asymptotic normality, we require additional assumptions.
\begin{assumption}
    \label{ass: KL}
    There exists a positive constant $M$ such that $\sup_j \mathbb{E} [(\eta_{j}^{(t)})^4] \leq M$.
\end{assumption}

\begin{assumption}
    \label{ass: smoothness}
    (i) Let $\tbPhi := \mathcal{D}^{\dagger} \bPhi$. For all $k_n$, $\| \tbPhi \| _{\L} < \infty$. 
    (ii) $\E\|\Dd\zeta_1\|^4 < \infty$.
\end{assumption}

\begin{assumption}
    \label{ass: bias of Pi}
   $\sqrt{n/k_n}~\brho (\widehat{\Pi}_{k_n} - \Pi_{k_n})(Y_{n+p}) \overset{\P}{\rightarrow} 0_{\H}$.
\end{assumption}

\begin{assumption}
    \label{ass: after k}
    $\sum_{j=k_n+1}^{\infty} \lambda_j \|\brho (\bnu_j)\|^2 = \O \left(1/n\right)$.
\end{assumption}

Assumption \ref{ass: KL} implies that $\E \|X_1\|^4 < \infty$, as required in Assumption \ref{ass: moment condition}. This condition holds for a broad class of real-valued random variables, including Gaussian and uniform distributions. Assumption \ref{ass: smoothness}(i) imposes a smoothness requirement on $\bPhi$, ensuring that it is at least as smooth as the operator $(\Dd)^{-1}=\sum_{j=1}^{k_n} ((\lambda_j^2+\alpha_n)/\lambda_j) \bnu_j\otimes\bnu_j$. Assumption \ref{ass: smoothness}(ii) restricts the noise level for technical feasibility.

Assumption \ref{ass: bias of Pi} is crucial for transforming a random bias into a non-random one, diverging from the result build in \citet{mas2007weak}. Proposition C.2 discusses the case when Assumption \ref{ass: bias of Pi} holds under mixing conditions (see \cite{dehling1983limit} for more details on mixing conditions). Lastly, Assumption \ref{ass: after k} ensures that $k_n$ is sufficiently large, preventing the truncation error from being excessively large.

Unlike common results, we first demonstrate that the asymptotic normality fails for the estimator $\hrho$.

\begin{theorem}
    \label{thm: CLT for operator}
    It is impossible for $\hrho - \brho$ to converge in distribution to a non-degenerate operator for the norm topology on $\K$, the space of compact operators from $\H^p$ to $\H$.
\end{theorem}

This conclusion is unsurprising. Establishing global asymptotic normality remains challenging even for i.i.d. data under nonparametric frameworks, and our dependent time series setting further complicates this objective.

Since the asymptotic normality for the operator estimator $\hrho$ does not hold, we consider the asymptotic normality for the predictor $\hrho Y_{n+p}$. Let
\begin{equation}
    \label{eq: t_n}
    t_n := \|\Dd \D\|_{\S}^2 = \sum_{j=1}^{k_n}\left(\frac{\lambda_j^2}{\lambda_j^2 + \alpha_n}\right)^2.
\end{equation}

\begin{theorem}
    \label{thm: CLT}
    (i) Under Assumptions \ref{ass: moment condition}-\ref{ass: smoothness}, if $\lambda_{k_n}^{-2}\alpha_n n^{1/2} = \O(1)$ and $n^{-1/2} k_n^{5/2} (\log k_n)^2 = \smallo(1)$, then
    \begin{equation}
        \sqrt{\frac{n}{t_n}} \left(\hrho Y_{n+p} - (\brho \widehat{\Pi}_{k_n})Y_{n+p}\right) \overset{d}{\rightarrow} G,
    \end{equation}
    where the convergence is in distribution in $\H$, and $G$ is an $\H$-valued Gaussian random variable with mean zero and covariance operator $\C_{\boldsymbol{\varepsilon}}$. \\
    (ii) Moreover, if Assumption \ref{ass: bias of Pi} holds, then
    \begin{equation}
        \sqrt{\frac{n}{t_n}} \left(\hrho Y_{n+p} - (\brho \Pi_{k_n})Y_{n+p}\right) \overset{d}{\rightarrow} G.
    \end{equation}
    (iii) Further, if Assumption \ref{ass: after k} holds, then
    \begin{equation}
        \sqrt{\frac{n}{t_n}} \left(\hrho Y_{n+p} - \brho Y_{n+p}\right) \overset{d}{\rightarrow} G.
    \end{equation}
\end{theorem}

\begin{remark}
    $G$ is said to be an $\H$-valued Gaussian random variable if its characteristic function satisfies 
    \begin{equation}
        \varphi_G(x) := \E \exp (\imath \langle x, G \rangle) = \exp(\imath \langle x, \E G \rangle - \langle \C_G (x), x \rangle / 2),
    \end{equation}
    where $\imath$ is the imaginary unit.
\end{remark}

Although $t_n$ is of the same order as $k_n$ under Assumption~\ref{ass: convexity of the function of eigenvalues}~\citep{cardot2007clt}, the normalization term $t_n$ is typically unknown in practice because it depends on the eigenvalues of the autocovariance operator $\D$. Hence, $t_n$ cannot be computed directly and must be approximated using estimated eigenvalues. To address this, we develop an adaptive version of Theorem~\ref{thm: CLT} in which each $\lambda_j$ is replaced by its empirical estimator $\hat{\lambda}_j$. The following corollary introduces a data-driven random normalization term that accounts for this substitution.

\begin{corollary}
    \label{cor: CLT}
    Theorem \ref{thm: CLT} still holds when $t_n$ is replaced by its empirical estimator $\widehat{t}_n$.
\end{corollary}

Specializing the Hilbert space $\H$ to the function space $L^2([0,1])$, we derive the specific form of the asymptotic normality for the predictor in terms of the weak convergence of the finite-dimensional distributions, and subsequently construct the confidence interval based on this result. For $X(s), s \in [0,1]$ belonging to $L^2([0,1])$, let $S = (s_1, \ldots, s_g)^\T \in [0,1]^g$ with $0 < g < \infty$ and denote the $g$-dimensional vector $X(S) = (X(s_1), \ldots, X(s_g))^\T$.
\begin{corollary}
    Under the conditions of Theorem \ref{thm: CLT}(iii), if $\H=L^2([0,1])$ , then
    \begin{equation}
        \sqrt{\frac{n}{t_n}} \left(\left(\hrho Y_{n+p} - \brho Y_{n+p}\right)(S)\right) \overset{d}{\rightarrow} N(\mathbf{0}, \Sigma_{\boldsymbol{\varepsilon}}),
    \end{equation}
    where $\Sigma_{\boldsymbol{\varepsilon}}$ is an $g \times g$ matrix with the $(i, j)$-entry $(\Sigma_{\boldsymbol{\varepsilon}})_{ij} = \E[\varepsilon_0(s_i)\varepsilon_0(s_j)]$. Further, for each $s_i \in S, i = 1,\ldots,g$, the $(1-\alpha)$ point-wise confidence interval for $(\brho Y_{n+p})(S)$ is given by
    \begin{equation}
        \label{eq: CI}
        \left[ (\hrho Y_{n+p})(s_i) \pm z_{1-\alpha/2} \sqrt{\frac{\widehat{t}_n}{n} (\widehat{\Sigma}_{\boldsymbol{\varepsilon}})_{ii}} \right],
    \end{equation}
    where $\widehat{\Sigma}_{\boldsymbol{\varepsilon}}$ is a consistent estimator of $\Sigma_{\boldsymbol{\varepsilon}}$, and $z_{1-\alpha/2}$ denotes the $(1-\alpha)$-th quantile of the standard normal distribution.
\end{corollary}

\subsection{Mean Squared Prediction Error}
\label{subsec: MSE}

To further characterize the predictive performance of the proposed estimator, we examine the mean squared prediction error (MSPE). Since the asymptotic normality results concern the prediction of future observations, the MSPE provides a natural measure of estimation accuracy. Although it exhibits a bias–variance decomposition similar to finite-dimensional settings, the sources of bias here arise from regularization and truncation, reflecting the infinite-dimensional structure of functional data.

\begin{theorem}
    \label{thm: MSE}
    If $n^{-1/2}k_n^{5/2} (\log k_n)^2 \rightarrow 0$, then the mean squared prediction error (MSPE) using the Yule-Walker estimator $\hrho$ can be decomposed as
    \begin{equation}
        \begin{aligned}
            &\E\|\hrho Y_{n+p} - \brho Y_{n+p}\|^2\\
            =& \O \Bigg(\underbrace{\sum_{j=1}^{k_n} \frac{\alpha_n^2 \lambda_j}{(\lambda_j^2+\alpha_n)^2} \|\brho (\bnu_j)\|^2}_{B_1: \text{regularization bias}} + \underbrace{\sum_{j=k_n+1}^{\infty} \lambda_j \|\brho (\bnu_j)\|^2}_{B_2: \text{truncation bias}}\Bigg) + \O \Bigg( \underbrace{\sum_{j=1}^{k_n}\left(\frac{\lambda_j^2}{\lambda_j^2 + \alpha_n}\right)^2 \frac{\sigma_{\varepsilon}^2}{n}}_{V: \text{estimation variance}} \Bigg).
        \end{aligned}
    \end{equation}
\end{theorem}

Theorem \ref{thm: MSE} reveals that the MSPE decomposes into three components: 
(i) Regularization bias ($B_1$), originating from Tikhonov regularization within the retained $k_n$-dimensional eigenspace; 
(ii) Truncation bias ($B_2$), induced by discarding eigenvalues beyond dimension $k_n$; 
(iii) Estimation variance ($V$), reflecting stochastic variability from finite-sample estimation.

Asymptotically, for a fixed truncation dimension $k_n$, increasing the regularization parameter $\alpha_n$ results in an increase in $B_1$, a constant $B_2$, and a decrease in $V$. When $\alpha_n = 0$, this corresponds to truncation without regularization, yielding $B_1 = 0$ and $V = n^{-1} k_n \sigma_{\varepsilon}^2$. As $\alpha_n \to \infty$, $B_1$ increases and converges to $\sum_{j=1}^{k_n} \lambda_j \|\brho (\bnu_j)\|^2$, and $V$ approaches zero. Similarly, fixing $\alpha_n$ and increasing $k_n$ leads to higher $B_1$ and $V$, but a reduction in $B_2$.

Thus, the choice of $\alpha_n$ and $k_n$ requires careful balancing between bias and variance. We recommend using cross-validation for selecting these parameters, as detailed in Section \ref{sec: real data}.

\section{Simulation Studies}
\label{sec: simulation}

In this section, we examine the finite sample performance of the proposed estimation method.
Given dimension $d$ and time lag $p$, we focus on the Hilbert space $\H=L^2([0,1], \R^d)$ and specify the kernels $R_j^{kl}(u,v)$ and $C_{\boldsymbol{\varepsilon}_r}(u,v)$ of the operators $\rho_j^{kl}$ and $\C_{\boldsymbol{\varepsilon}_r}$ for $j=1,\ldots,p$ and $k,l,r=1,\ldots,d$. By K-L expansion, $\varepsilon_{tr}(u)$ can be expressed as
\begin{equation}
    \label{eq: KL of eps}
    \varepsilon_{tr}(u) = \sum_{j=1}^{\infty} \eta_j^{(tr)} \sqrt{\xi_j^{(r)}} \boldsymbol{e}_j^{(r)}(u),
\end{equation}
where $(\xi_j^{(r)},\boldsymbol{e}_j^{(r)}(u)), j\in\bN^+$ are the eigenpairs of $C_{\boldsymbol{\varepsilon}_r}(u,v)$ and $\eta_j^{(tr)}\overset{i.i.d.}{\sim} N(0,1)$.

We generate $n$ samples $\boldsymbol{X}_1(u),\ldots,\boldsymbol{X}_n(u)$. To maintain the stationarity, we generate $3n$ samples in $L^2([0,1], \R^d)$ using equation (\ref{eq: special ARHp})-(\ref{eq: kernel operator}), initializing the first $p$ samples as zero and retaining only the last $n$ samples. Given $p$, we apply the proposed method to obtain the estimations $\widehat{R}_j^{kl}(u,v)$ and $\widehat{C}_{\boldsymbol{\varepsilon}_r}(u,v)$. The replications are 200 times.

We employ two measures to evaluate the estimations and prediction. The first is the integrated mean squared error (IMSE) for the kernel of each $\widehat{\rho}_j^{kl}$, defined as $\E[\iint(\widehat{R}_j^{kl}(u,v)-R_j^{kl}(u,v))^2dudv]^{1/2}$. The second measure is the prediction error (PE) for forecasting $m \in \bN^+$ upcoming samples, defined as $(dm)^{-1}\sum_{r=1}^{d}\sum_{v=1}^{m}\E \int(\widehat{x}_{n+p+v,r}^*(u)-x_{n+p+v,r}^*(u))^2du$, where $x^*(u)$ represents the true samples to be predicted. In practice, we set $m=10$ and each function is recorded at $g=101$ equidistant grid points on $[0,1]$. We conduct two examples inspired by \citet{rubin2020spectral}.

\noindent\textbf{Example 1.} Consider the case where $d=1,p=4$ with the autoregressive operators $\rho_1, \rho_2, \rho_3, \rho_4$ and the innovation covariance operator $\C_{\boldsymbol{\varepsilon}}$ defined as integral operators with the respective kernels
\begin{align*}
    & R_1(u,v) = 0.3\sin(u-v), \
    R_2(u,v) = 0.3\cos(u-v),\\
    & R_3(u,v) = 0.3\sin(2u), \
    R_4(u,v) = 0.3\cos(v),\\
    & C_{\boldsymbol{\varepsilon}}(u,v) = \min(u,v),
\end{align*}
for $u,v\in[0,1]$. Note that $C(u,v) = \min(u,v)$ is the covariance kernel of the standard Brownian motion on $[0,1]$, which admits the decomposition (\citealp{deheuvels2003karhunen})
\begin{equation*}
    C(u,v) = \sum_{j=1}^{\infty} \frac{1}{((j-0.5)\pi)^2} \sqrt{2}\sin\left((j-0.5)\pi u\right) \sqrt{2}\sin\left((j-0.5)\pi v\right).
\end{equation*}
For $n=100,200,500,1000$, we apply the proposed methods to generate data $X_1,\ldots, X_n$ and conduct parameter estimation and prediction.

Figure \ref{fig: ARH_4_line_plot} illustrates PE versus Tikhonov parameters $\alpha_n$ and $k_n$ ($n=100$, averaged over 200 replications). When $\alpha_n=0$, PE initially decreases and then increases as $k_n$ increases, reaching a minimum value of 0.031 at $k_n=3$. However, the cumulative contribution of the first three eigenvalues is only 70\%, which is insufficient to capture all the information. When $\alpha_n>0$, PE converges to a common constant as $\alpha_n$ increases regardless of $k_n$. For $k_n\leq3$ (low cumulative contribution), adding $\alpha_n$ fails to reduce PE. For $k_n=5$ (85\% contribution), PE first decreases then increases with $\alpha_n$. These behaviors align with Theorem \ref{thm: MSE}. The optimal parameters minimizing PE are $k_n=10$ and $\alpha_n=0.04$, yielding PE=0.028 with 94\% cumulative contribution. Consequently, in the simulation study, we determine the optimal values $\alpha_n$ and $k_n$ that minimize PE.

\begin{figure}[htbp]
    \centering
    \subfloat[]{ 
        \includegraphics[width=0.5\textwidth]{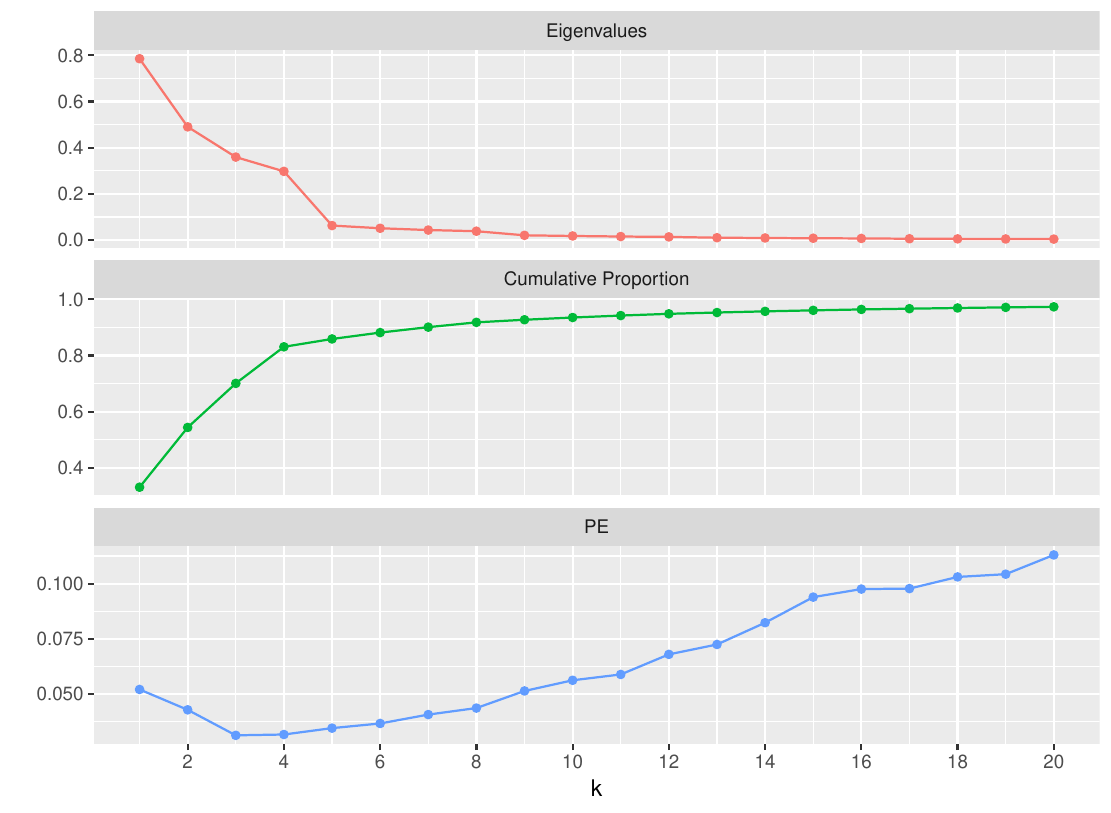}
    }
    \subfloat[]{
        \includegraphics[width=0.5\textwidth]{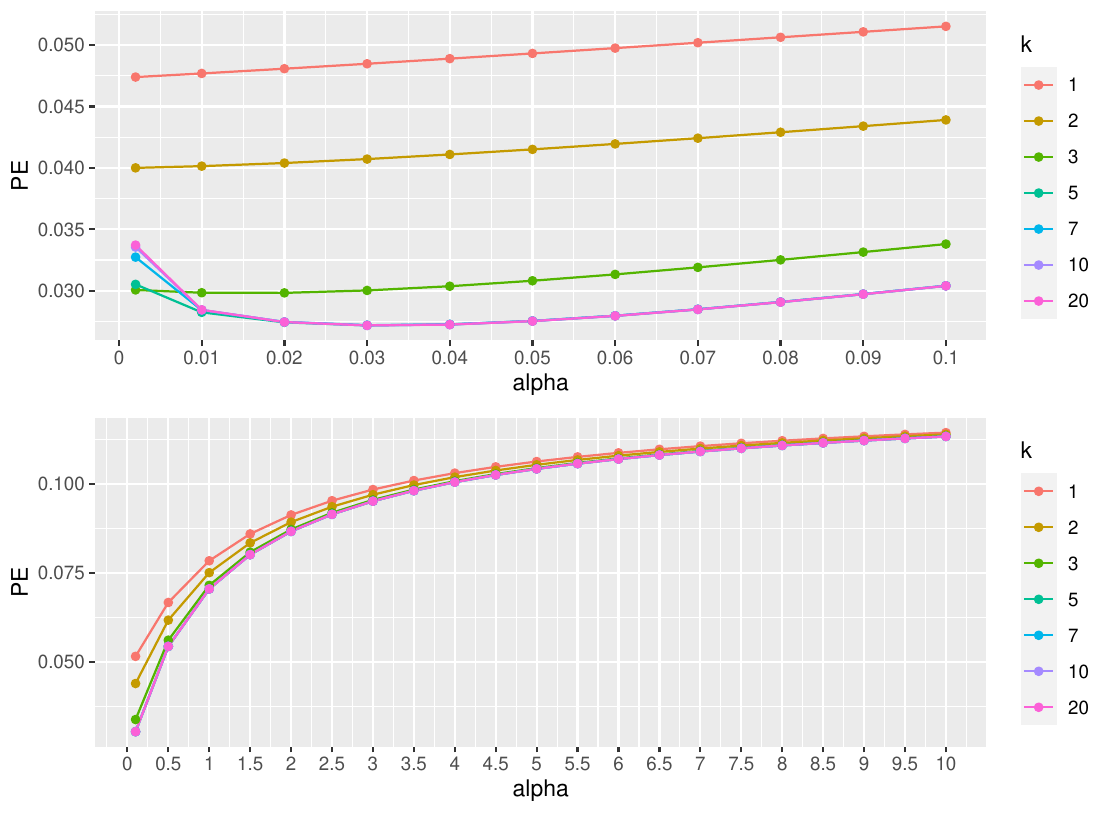}
    }
    \caption{(a) \textbf{Top} and \textbf{Middle}: The first 20 eigenvalues of $\hD$ and their cumulative proportion under $n=100$. \textbf{Bottom}: PE on $k_n$ when $\alpha_n=0, n=100$. (b) PE on $\alpha_n$ for different $k_n$ when $n=100$.}
    \label{fig: ARH_4_line_plot}
\end{figure}

Table \ref{tab: ARH_4} demonstrates the consistency of both the estimators and predictor. As $n$ increases, both PE and IMSE decrease, confirming the convergence of the proposed method. Figure \ref{fig: ARH_4_qq_plot} further illustrates the asymptotic normality of the predictor. To assess normality, we examine the QQ-plots of randomly projected prediction errors with $\boldsymbol{u}^\T (\widehat{\widetilde{\boldsymbol{X}}}_{n+p+1} - \widetilde{\boldsymbol{X}}_{n+p+1})$, where $\boldsymbol{u}$ is a 101-dimensional random vector drawn from a multivariate uniform distribution. The projected errors are standardized before constructing the QQ-plots. These findings offer empirical validation of Theorem \ref{thm: consistency} and Theorem \ref{thm: CLT}.

\begin{table}[htbp]
    \caption{PE and IMSE for kernels $\rho_1, \rho_2, \rho_3, \rho_4$. The standard errors are in parentheses.}
    \centering
    \begin{tabular}{ccccccc}
    \toprule
        & $n$  & PE & IMSE1 & IMSE2 & IMSE3 & IMSE4  \\ \midrule
        & 100  & 0.0280(0.058) & 0.1125(0.034) & 0.1462(0.049) & 0.1374(0.042) & 0.1560(0.035)  \\
        & 200  & 0.0112(0.015) & 0.0974(0.029) & 0.1183(0.027) & 0.1134(0.027) & 0.1341(0.026)  \\
        & 500  & 0.0058(0.009) & 0.0781(0.022) & 0.1014(0.020) & 0.0984(0.021) & 0.1165(0.022)  \\
        & 1000  & 0.0033(0.005) & 0.0683(0.021) & 0.0875(0.019) & 0.0864(0.019) & 0.1026(0.022)  \\ \bottomrule
    \end{tabular}
    \label{tab: ARH_4}
\end{table}

\begin{figure}[htbp]
    \centering
    \includegraphics[width=1\textwidth]{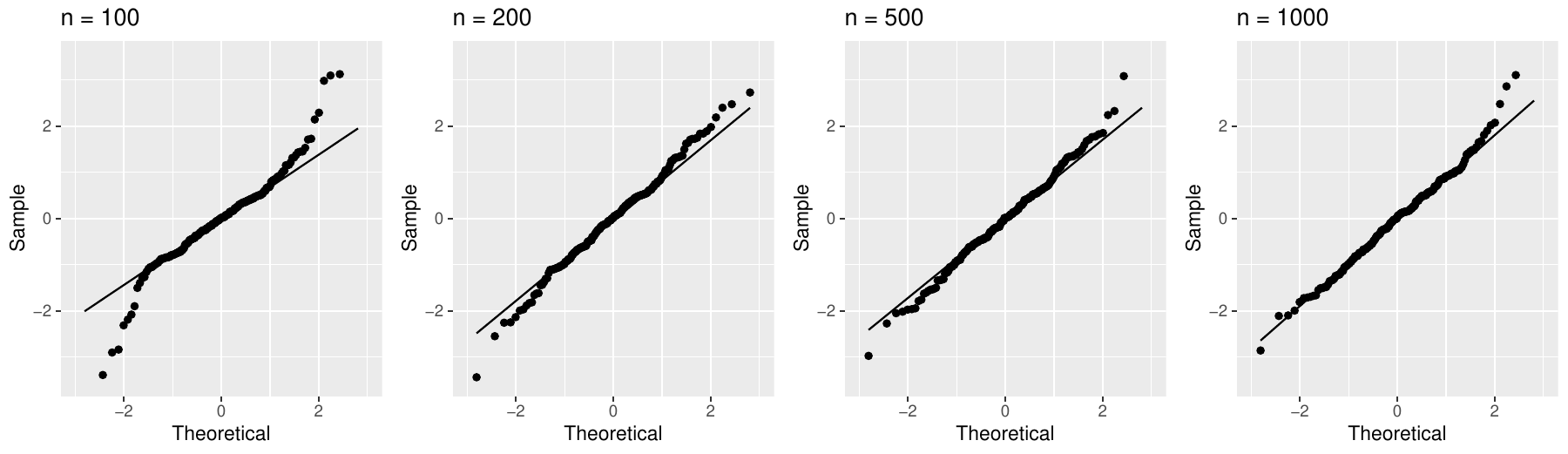}
    \caption{QQ-plots of $\boldsymbol{u}^\T (\widehat{\widetilde{\boldsymbol{X}}}_{n+p+1} - \widetilde{\boldsymbol{X}}_{n+p+1})$.}
    \label{fig: ARH_4_qq_plot}
\end{figure}

\begin{figure}[htbp]
    \centering
    \includegraphics[width=0.95\textwidth]{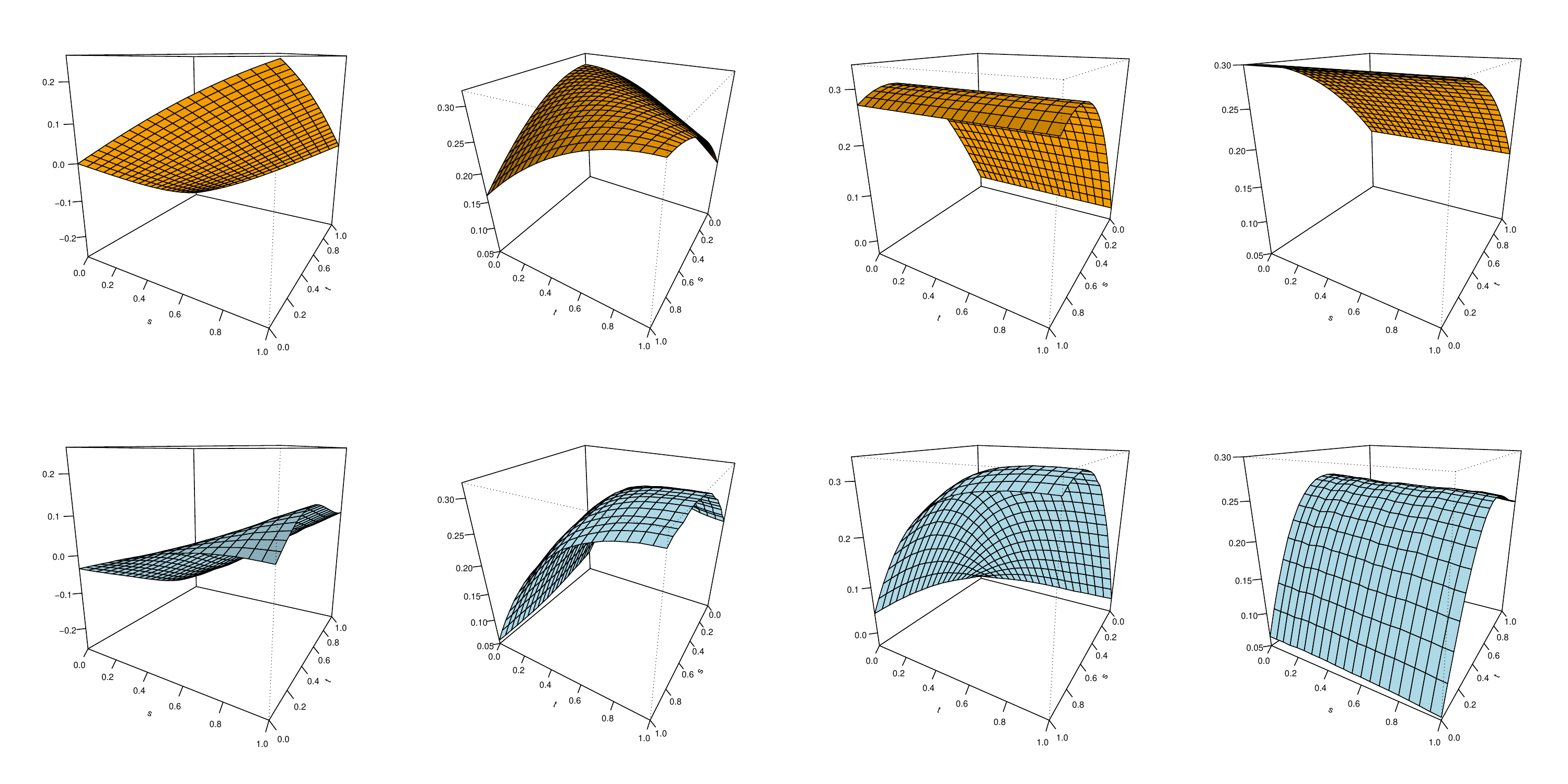}
    \caption{\textbf{Top}: The true kernels of the autoregressive operators $\rho_1, \rho_2, \rho_3, \rho_4$. \textbf{Bottom}: The estimated kernels of autoregressive operators $\widehat{\rho}_1, \widehat{\rho}_2, \widehat{\rho}_3, \widehat{\rho}_4$ when $n = 1000$.}
    \label{fig: ARH_4_rho_plot}
\end{figure}

\begin{figure}[htbp]
    \centering
    \includegraphics[width=0.85\textwidth]{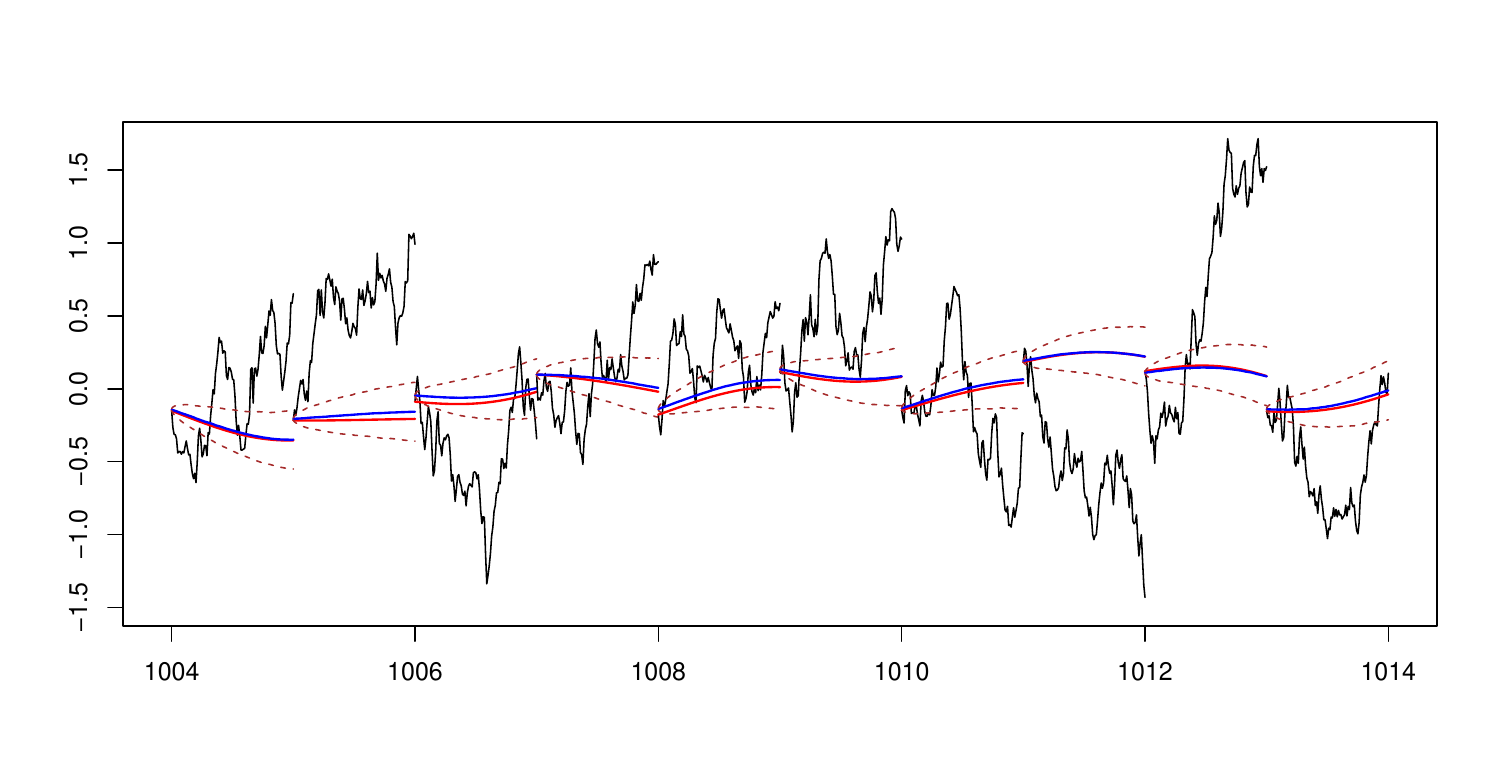}
    \caption{\footnotesize Prediction of $X_{n+p+v}$ given $X_{n+p+v-1}, \ldots, X_{n+v}$ for $v=1,2,\ldots,10$ when $n = 1000$. Within each interval $[t-1,t]$, the black solid line represents the generated data $X_t$, the red solid line denotes $\brho Y_{t-1}$ excluding white noise, the blue solid line corresponds to the prediction $\hrho Y_{t-1}$, and the brown dashed lines denote the 95\% confidence interval for $\brho Y_{t-1}$.}
    \label{fig: ARH_4_pred}
\end{figure}

Figure \ref{fig: ARH_4_rho_plot} presents the surfaces formed by the kernels of both the true and estimated autoregressive operators. Despite some deviation at the boundaries, the overall trend of the estimated surfaces closely aligns with their true counterparts. Figure \ref{fig: ARH_4_pred} displays the prediction results for the unobserved samples, demonstrating strong predictive performance. Notably, the underlying curves, excluding white noises, fall within the constructed confidence intervals, further supporting the validity of the proposed estimation and inference methods.

\noindent\textbf{Example 2.} Consider the case where $d=2, p=2$ with the autoregressive operators $\rho^{11}, \rho^{12}, \rho^{21}, \rho^{22}$ and the innovation covariance operators $\C_{\boldsymbol{\varepsilon}_{11}}, \C_{\boldsymbol{\varepsilon}_{12}}$ defined as integral operators with the respective kernels
\begin{align*}
    & R^{11}(u,v) = 0.3\exp\{(u^2+v^2)/2\}, \
    R^{12}(u,v) = 0.3(u+v),\\
    & R^{21}(u,v) = 0.9(u-v), \
    R^{22}(u,v) = 0.3\sin(2v),\\
    & C_{\boldsymbol{\varepsilon}_{11}}(u,v) =C_{\boldsymbol{\varepsilon}_{12}}(u,v) = \min(u,v) - uv
\end{align*}
for $u,v\in[0,1]$. Note that $C(u,v) = \min(u,v) - uv$ is the covariance kernel of the Brownian bridge on $[0,1]$, which admits the decomposition (\citealp{deheuvels2003karhunen})
\begin{equation*}
    C(u,v) = \sum_{j=1}^{\infty} \frac{1}{(j\pi)^2} \sqrt{2}\sin(j\pi u) \sqrt{2}\sin(j\pi v).
\end{equation*}
For $n=100,200,500,1000$, we employ the analogous process as Example 1 and obtain the following results: consistency results are tabulated in Table \ref{tab: ARH^2}, while other figures are presented in Appendix D. Unlike Example 1, Example 2 generates two curve sequences exhibiting a dependent structure. Similar to the analysis in Example 1, the results of Example 2 further confirm that our estimation method and theoretical findings remain valid for multivariate functional time series.

\begin{table}[htbp]
    \caption{PE and IMSE for kernels $\rho^{11}, \rho^{12}, \rho^{21}, \rho^{22}$. The standard errors are in parentheses.}
    \centering
    \begin{tabular}{ccccccc}
    \toprule
        & $n$ & PE & IMSE(1,1) & IMSE(1,2) & IMSE(2,1) & IMSE(2,2)  \\ \midrule
        & 100  & 0.0039(0.005) & 0.2114(0.043) & 0.1937(0.037) & 0.2706(0.035) & 0.1771(0.039)  \\
        & 200  & 0.0021(0.002) & 0.1785(0.029) & 0.1624(0.022) & 0.2517(0.026) & 0.1484(0.032)  \\
        & 500  & 0.0012(0.001) & 0.1532(0.024) & 0.1474(0.018) & 0.2264(0.020) & 0.1227(0.022)  \\
        & 1000  & 0.0008(0.001) & 0.1358(0.021) & 0.1288(0.015) & 0.2037(0.016) & 0.1126(0.020)  \\ \bottomrule
    \end{tabular}
    \label{tab: ARH^2}
\end{table}

\section{Real Data Analysis}
\label{sec: real data}

In this section, we illustrate the proposed method using a publicly available wearable sensor dataset from \url{https://ubicomp.eti.uni-siegen.de/home/datasets/ubicomp12/}; see \cite{BerlinEugen2012Dlaw} for details. The dataset consists of raw acceleration measurement in three dimensions recorded by a wrist-worn sensor worn by six participants, each performing one leisure activity—such as badminton, cycling, flamenco dancing, playing guitar, gym workouts, or Zumba aerobics for 30-90 minutes daily. The acceleration data, as illustrated in Figure \ref{fig: guitar_3}, is naturally considered as multivariate functional data. We apply the proposed method to estimate the autoregressive operators for each activity, aiming to uncover intrinsic movement patterns and provide new insights into activity pattern recognition.

Due to the highly irregular sampling intervals in the original data, we first apply linear spline interpolation for equally spaced observations, then cubic spline interpolation for smoothing. The data is segmented into 1-second intervals, each containing 101 points. Each segment can be viewed as a realization of a random element in the space $L^2([0,1], \R^3)$. For each activity, we use the first 90\% of the data as the training set and reserve the remaining 10\% for testing. Applying the methodology proposed in Section \ref{sec: example}, we estimate the autocovariance and autoregressive operators. As discussed in Section \ref{subsec: MSE}, the choice of parameters $\alpha_n$ and $k_n$ plays a critical role in estimation accuracy. We employ 5-fold cross-validation for selection, adopting the partitioning as illustrated in Figure \ref{fig: cv} to handle functional time series dependence.

We fit an FAR(5) model to the data. Figure \ref{fig: rho1} presents the estimated kernels of $\widehat{\brho}_1$ for six different leisure activities, revealing distinct autoregressive patterns in both magnitude and shape. To further illustrate the differences among activities, we extract the first two principal components of the vectorized $\widehat{\brho}_1$ and apply the Support Vector Machine (SVM) with a linear kernel to determine the separation hyperplane, as shown in Figure \ref{fig: PCA}. The results indicate effective clustering of different activities, with the separation hyperplane successfully partitioning the feature space.

Finally, using the estimated autoregressive operators, we predict the data in the testing set and construct a 95\% confidence interval. Figure D.4 displays the prediction results for the activity \enquote{cycling\_2}, showing that the model effectively captures the trend of the activity. These findings demonstrate that the proposed model and estimation method perform well on real-world wearable sensor data.

\begin{figure}[htbp]
    \centering
    \centering
        \begin{tabular}{cc}
            \includegraphics[width=0.8\textwidth]{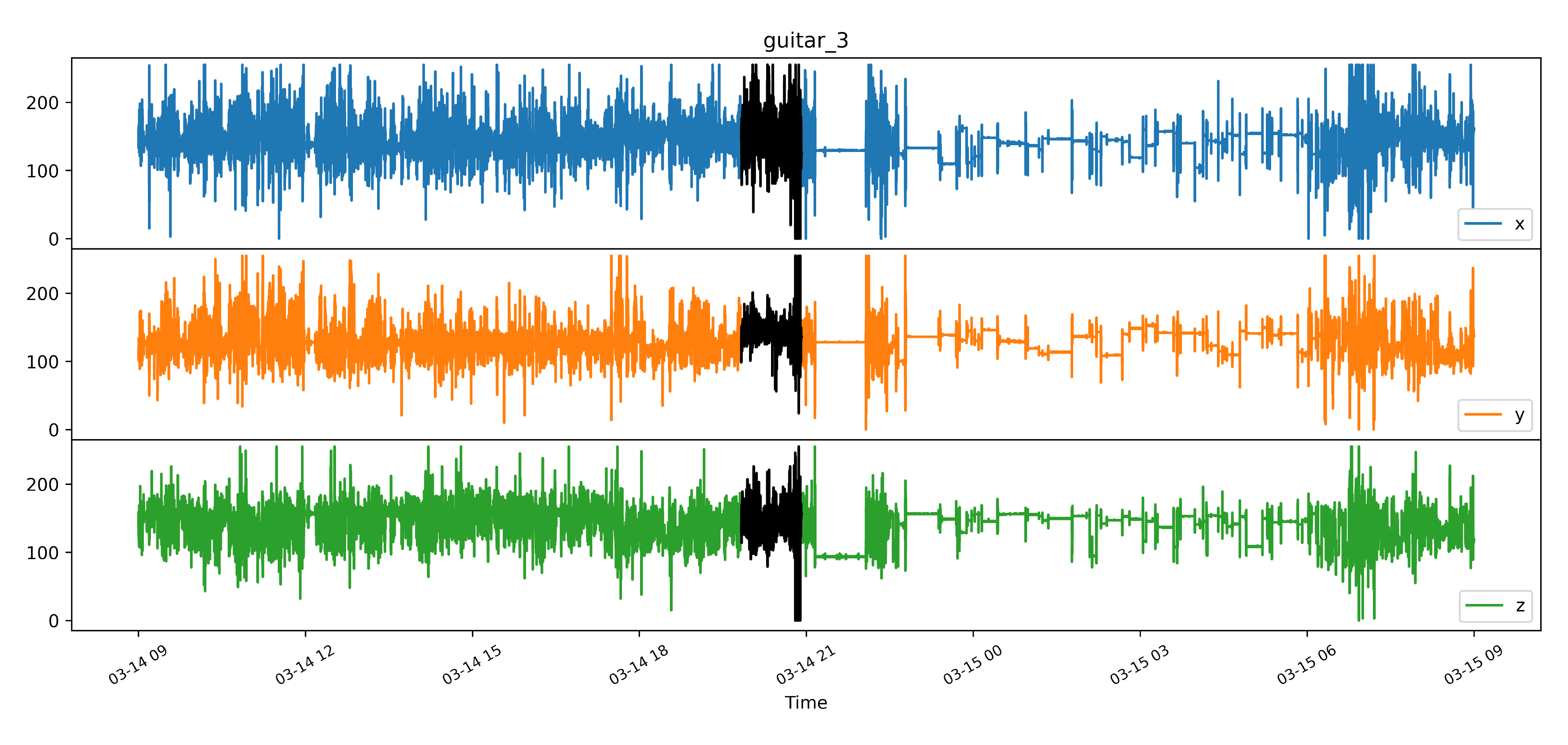}\\
            \hspace{2mm}{\includegraphics[width=0.8\textwidth]{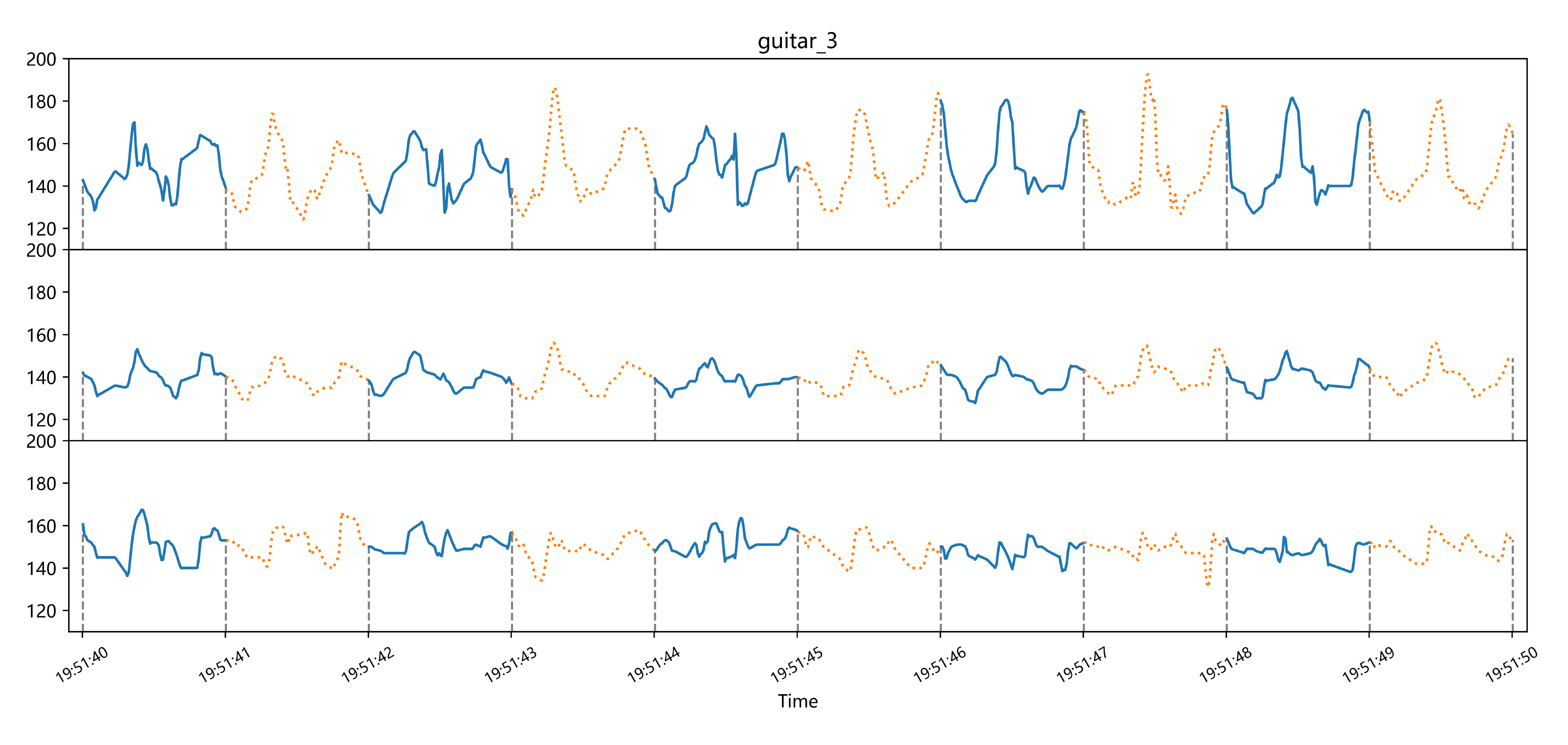}}
        \end{tabular}
    \caption{\small \textbf{Top}: The 24-hour, 3-dimensional acceleration data collected from the wrist-worn sensor of the participant who chose playing the guitar as his leisure activity. The black segments mark activity periods. \textbf{Bottom}: A zoomed-in view of the acceleration data for a selected 10-second period during the leisure activity, with the data divided into 1-second segments.}
    \label{fig: guitar_3}
\end{figure}

\begin{figure}
    \centering
    \includegraphics[height=18cm]{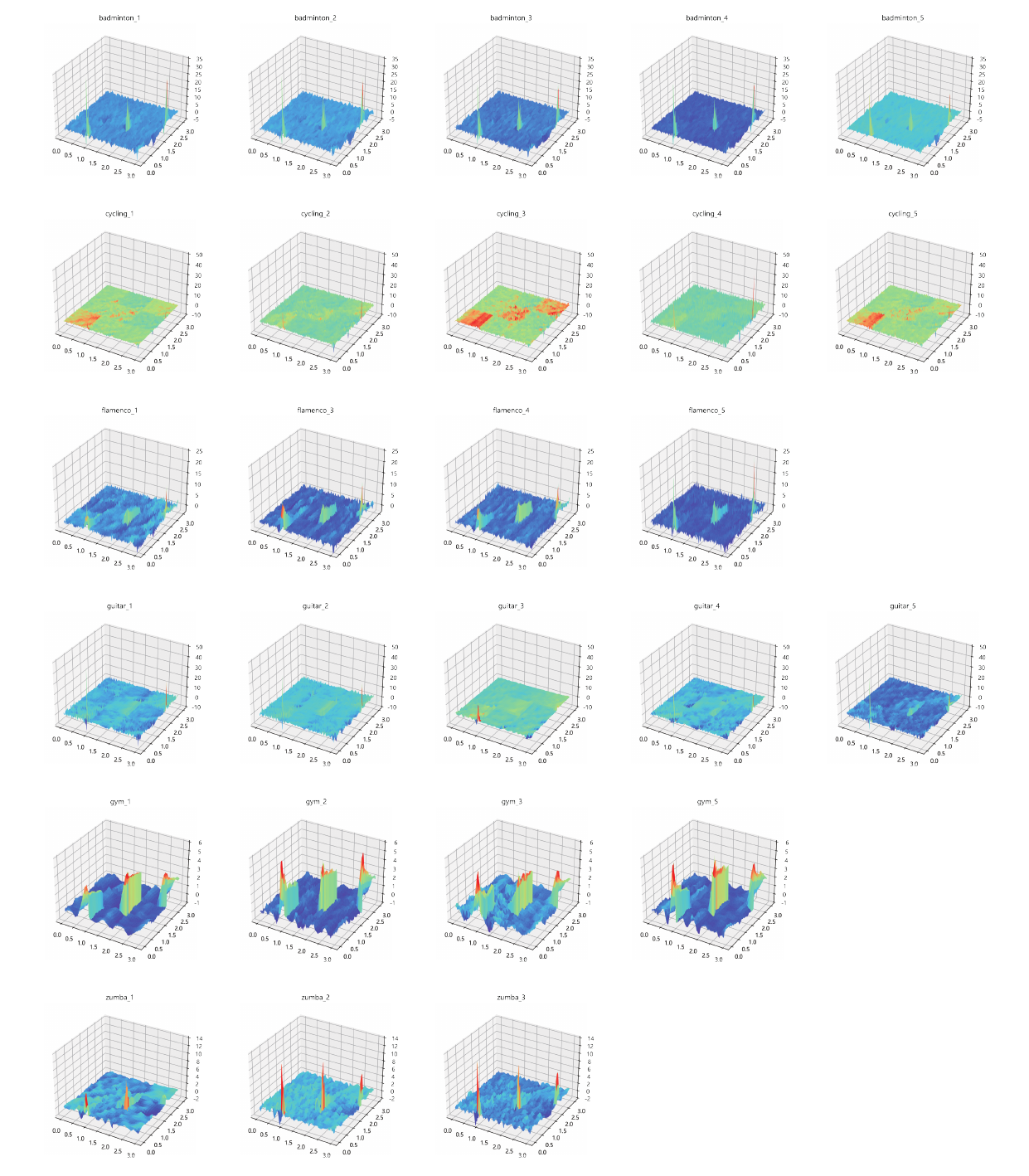}
    \caption{\small The estimated kernels of autoregressive operators $\widehat{\brho}_1$ for six different activities. The same row but different columns represent the estimations of the same activity on different days.}
    \label{fig: rho1}
\end{figure}

\begin{figure}
    \centering
    \includegraphics[width=0.8\linewidth]{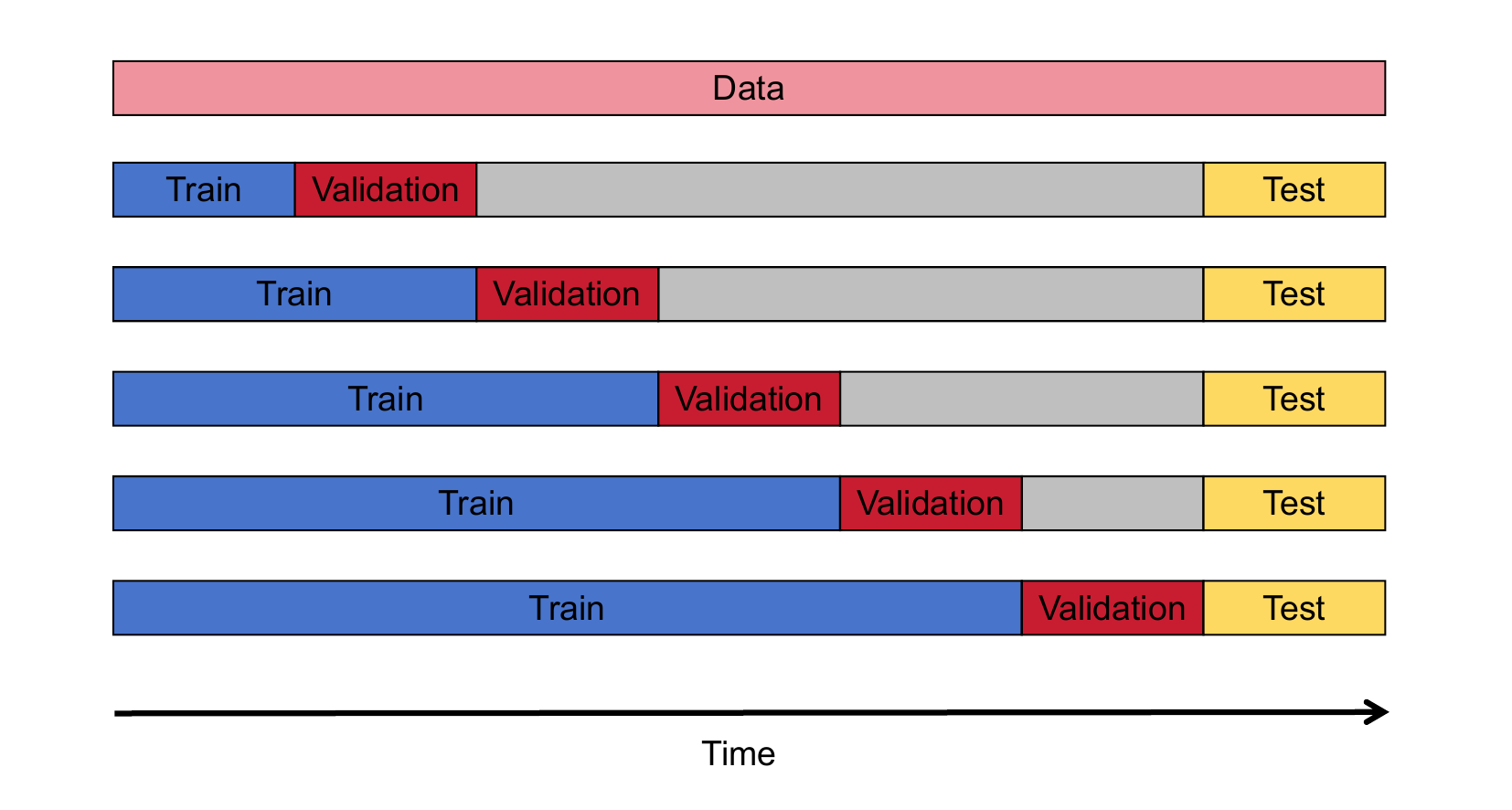}
    \caption{The 5-fold cross-validation diagram for the functional time series.}
    \label{fig: cv}
\end{figure}

\begin{figure}[htbp]
    \centering
        \begin{tabular}{ccc}
            \includegraphics[width=0.5\textwidth]{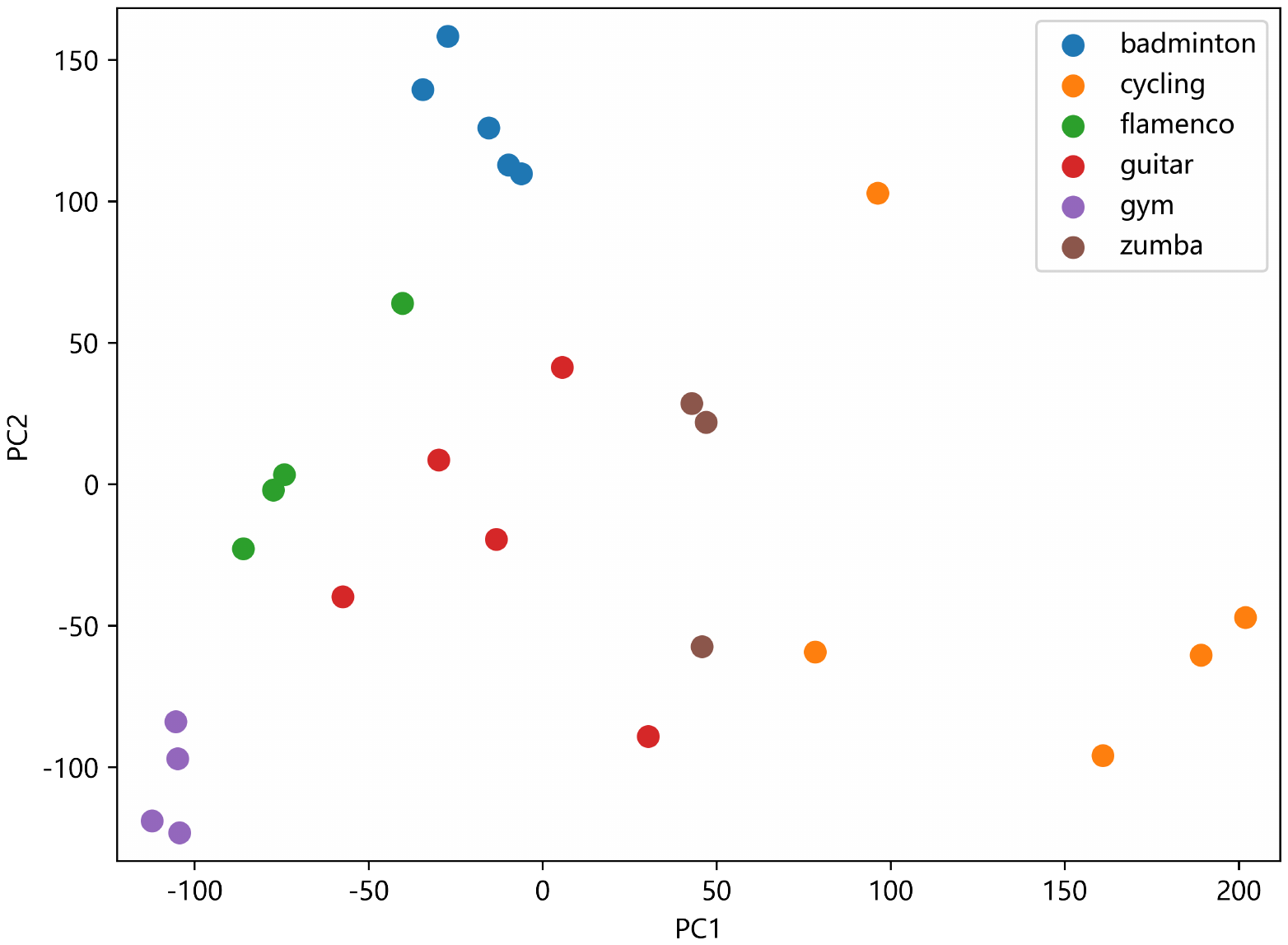}
            \includegraphics[width=0.5\textwidth]{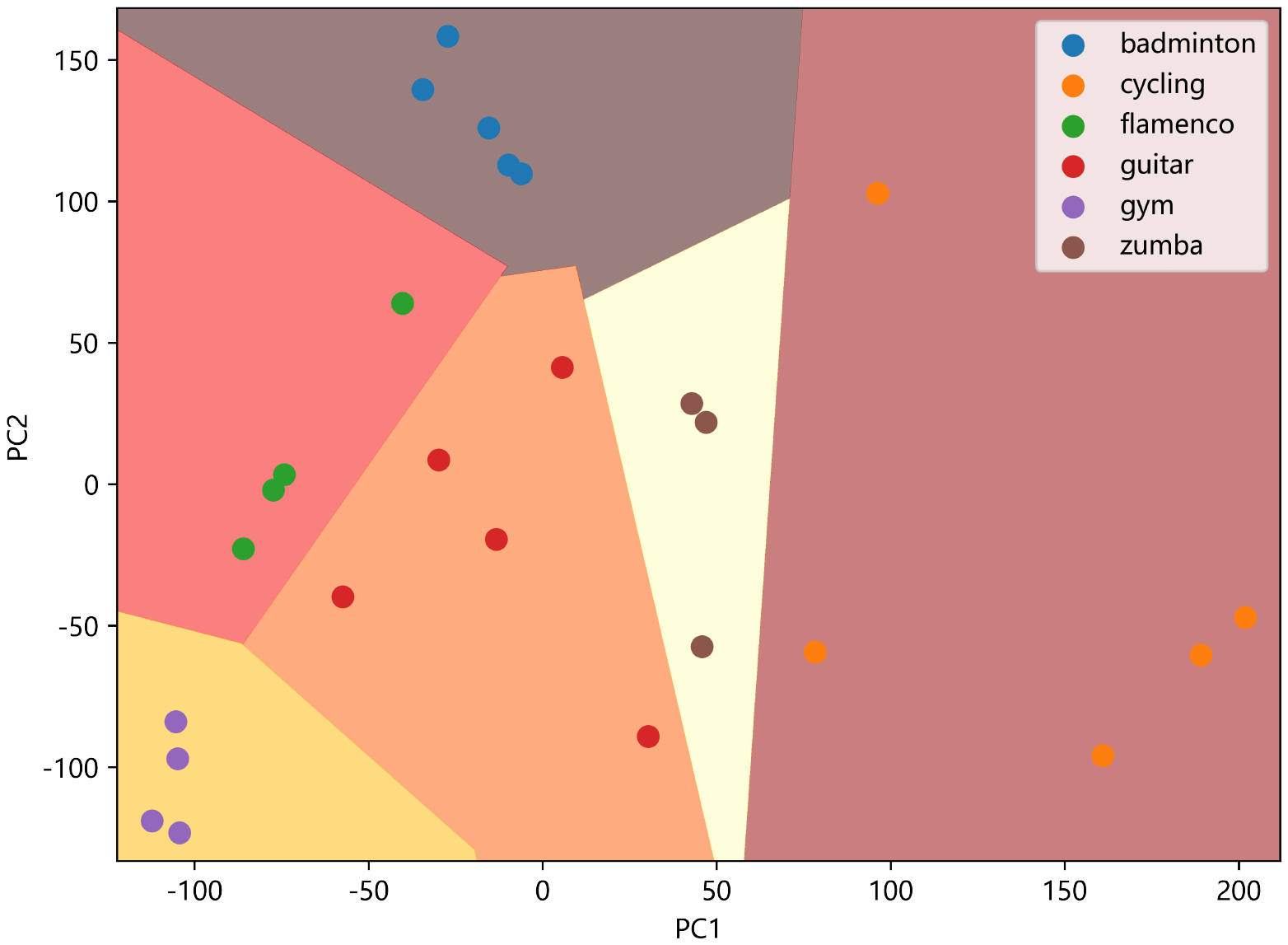}
        \end{tabular}
    \caption{\small \textbf{Left}: The scatter plot of the first and second principal components of vectorized $\widehat{\brho}_1$, where the points in various colors denote various leisure activities. \textbf{Right}: The segmentation hyperplanes obtained using SVM with the linear kernel.}
    \label{fig: PCA}
\end{figure}

\section{Conclusion}
\label{sec: conclusion}
This paper addresses the higher-order FAR($p$) model in functional time series, overcoming the fundamental challenges of infinite-dimensional inference through a novel regularization framework. By establishing the Yule-Walker equations, we develop regularized estimators that bridge the gap between finite and infinite-dimensional settings. Our theoretical contributions include proving the consistency of the estimators and establishing asymptotic normality for predictors with mean squared prediction error analysis. The concrete implementation in $L^2$ space enables practical applications, validated through simulations and a wearable sensor study that demonstrates activity recognition via estimated autoregressive operators. This paper does not address the estimation of the order $p$, and references to \cite{kokoszka2013determining} are suggested for further exploration.


\bibliography{bibliography}

\end{document}